\documentclass[letterpaper,11pt]{article}
\usepackage{amsmath, amssymb}
\usepackage{color}
\usepackage{fullpage}
\usepackage[numbers]{natbib}
\usepackage[noend]{algorithmic}
\usepackage{tikz}
\usepackage{enumerate}
\usepackage{thm-restate}
\usepackage{graphicx}
\usepackage{amsopn}
\usepackage{caption}
\usepackage{hyperref}
\usepackage{subcaption}
\usepackage{zerosum,csquotes}
\usepackage{enumitem,linegoal}
\usepackage[linesnumbered,ruled,vlined]{algorithm2e}
\usepackage{comment}
\usepackage{ctable}

\usepackage{mathtools}
\usepackage[flushleft]{threeparttable}
\usepackage[normalem]{ulem}

\usepackage{footnote}
\makesavenoteenv{tabular}
\makesavenoteenv{table}

\usepackage[margin=1in]{geometry}

 \newtheorem{theorem}{Theorem}
 \newtheorem{conj}[theorem]{Conjecture}
 \newtheorem{lemma}[theorem]{Lemma}
 
 \newtheorem{corollary}[theorem]{Corollary}

\makeatletter
\def\GrabProofArgument[#1]{ #1: \egroup\ignorespaces}
\def\proof{\noindent\textbf\bgroup Proof%
    \@ifnextchar[{\GrabProofArgument}{. \egroup\ignorespaces}}

\makeatother

\definecolor{mygreen}{RGB}{20,140,80}
\definecolor{mydarkgray}{gray}{0.15} 
\hypersetup{
     colorlinks=true,
     citecolor= mygreen,
     linkcolor= black
}

\newcommand*\samethanks[1][\value{footnote}]{\footnotemark[#1]}
\newcommand{\citeboth}[1]{\hypersetup{citecolor=mydarkgray}\citeauthor{#1}\hypersetup{citecolor=mygreen} \cite{#1}}
\definecolor{burntorange}{rgb}{0.8, 0.33, 0.0}
\definecolor{figurePink}{rgb}{1.0, 0.20, 0.74}
\definecolor{figureBlue}{rgb}{0.03, 0.27, 0.71}

\newcommand{\Alg}{\mathcal{A}}
\renewcommand{\log}{\lg}
\newcommand{\dist}{\textrm{dist}}
\newcommand{\Fam}{\mathcal{F}}
\newcommand{\E}{\mathbb{E}}
\newcommand{\R}{\mathbb{R}}
\newcommand{\In}{\textrm{In}}
\newcommand{\eps}{\varepsilon}

\title{Lower Bounds for External Memory Integer Sorting \\ via Network Coding}
\author{Alireza Farhadi \thanks{\texttt{\{farhadi, hajiagha\}@cs.umd.edu}. Supported in part by NSF CAREER award CCF-1053605, NSF AF:Medium grant CCF-1161365, NSF BIGDATA grant IIS-1546108, NSF SPX grant CCF-1822738,
 UMD AI in Business and Society Seed Grant and UMD Year of Data Science Program Grant.}\\University of Maryland\\ College Park, MD
\and MohammadTaghi Hajiaghayi \samethanks[1] \\ University of Maryland\\
College Park, MD\and Kasper Green Larsen
\thanks{\texttt{larsen@cs.au.dk}. Supported by a Villum
  Young Investigator Grant and an AUFF Starting Grant.}\\ Aarhus
University\\ Denmark \and Elaine
Shi\thanks{\texttt{runting@gmail.com}. Supported in part by NSF under grant number CCF1822805.}\\Cornell University\\
Ithaca, NY}
\date{}
\begin{document}

\maketitle

\begin{abstract}
Sorting extremely large datasets is a frequently
occuring task in practice. These datasets are usually much
larger than the computer's main memory; thus external memory sorting
algorithms, first introduced by \citeboth{aggarwal1988input} (1988),
 are often used. The complexity of comparison based external memory
 sorting has been understood for decades by now, however the situation
 remains elusive if we assume the keys to be sorted are integers. In
 internal memory, one can sort a set of $n$ integer keys of
 $\Theta(\lg n)$ bits each in
 $O(n)$ time using the classic Radix Sort algorithm, however in
 external memory, there are no faster integer sorting algorithms known than the simple
 comparison based ones. Whether such algorithms exist has remained a
 central open problem in external memory algorithms for more than
 three decades.

In this paper, we present
a {\em tight} conditional lower bound on the complexity of external
memory sorting of integers. Our lower
bound is based on a famous conjecture in network coding by \citeboth{lili}, who conjectured that network coding cannot help anything
beyond the standard multicommodity flow rate in undirected graphs.

The only previous work connecting the Li and Li conjecture to lower
bounds for algorithms is due to \citeboth{Adler:soda}.
Adler et al. indeed obtain relatively simple lower bounds for
\emph{oblivious} algorithms (the memory access pattern is fixed
and independent of the input data). Unfortunately obliviousness is a
strong limitations, especially for integer sorting: we show that the
Li and Li conjecture implies an $\Omega(n \log n)$ lower bound for internal memory \emph{oblivious}
sorting when the keys are $\Theta(\lg n)$ bits. This is
in sharp contrast to the classic (non-oblivious) Radix Sort
algorithm. Indeed going beyond obliviousness is
highly non-trivial;  we need to introduce several new methods and
involved techniques, which are of their own interest, to obtain our
tight
lower bound for external memory integer sorting.
\end{abstract}

\thispagestyle{empty}
\newpage

\section{Introduction}
Sorting is one of the most basic algorithmic primitives and has attracted lots of attention from the beginning of the computing era.
Many classical algorithms have been designed for this problem such as Merge Sort, Bubble Sort, Insertion Sort, etc. As sorting extremely large data has become essential for many applications, there has been a strong focus on designing more efficient algorithms for sorting big datasets \cite{aggarwal1988input}. These datasets are often much larger than the computer's main memory and the performance bottleneck changes from the being the number of CPU instructions executed, to being the number of accesses to slow secondary storage. In this external memory setting, one usually uses the external memory model to analyse the performance of algorithms. External memory algorithms are designed to minimize the number of \textit{input/output (I/O)}s between the internal memory and external memory (e.g. hard drives, cloud storage, etc.), and we measure the complexity of an algorithm in terms of the number of I/Os it performs.

Formally, the external memory model consists of a main memory that can hold $M$ words of $w$ bits each (the memory has a total of $m = Mw$ bits), and an infinite (random access) disk partitioned into blocks of $B$ consecutive words of $w$ bits each (a block has a total of $b = Bw$ bits). The input to an external memory algorithm is initially stored on disk and is assumed to be much larger than $M$. An algorithm can then read blocks into memory, or write blocks to disk. We refer jointly to these two operations as an I/O. The complexity of an algorithm is measured solely in terms of the number of I/Os it makes.

\citeboth{aggarwal1988input} considered the sorting problem in the external memory model. A simple modification to the classic Merge Sort algorithm yields a comparison based sorting algorithm that makes $O((n/B)\lg_{M/B}(n/B))$ I/Os for sorting an array of $n$ comparable records (each storable in a word of $w$ bits). Notice that $O(n/B)$ would correspond to \emph{linear} I/Os, as this is the amount of I/Os needed to read/write the input/output. \citeboth{aggarwal1988input} complemented their upper bound with a matching lower bound, showing that comparison based external memory sorting algorithms must make $\Omega((n/B)\lg_{M/B}(n/B))$ I/Os. In the same paper, Aggarwal and Vitter also showed that any algorithm treating the keys as \emph{indivisible} atoms, meaning that keys are copied to and from disk blocks, but never reconstructed via bit tricks and the like, must make $\Omega(\min\{n,(n/B)\lg_{M/B}(n/B) \})$ I/Os. This lower bound does not assume a comparison based algorithm, but instead makes an indivisibility assumption. Notice that the lower bound matches the comparison based lower bound for large enough $B$ ($B > \lg n$ suffices). The comparison and indivisibility settings have thus been (almost) fully understood for more than three decades.

However, if the input to the sorting problem is assumed to be $w$ bit integers and we allow arbitrary manipulations of the integers (hashing, XOR tricks etc.), then the situation is completely different. In the standard internal memory computational model, known as the word-RAM, one can design integer sorting algorithms that far outperform comparison based algorithms regardless of $w$. More concretely, if the word size and key size is $w= \Theta(\lg n)$, then Radix Sort solves the problem in $O(n)$ time, and for arbitrary $w$, one can design sorting algorithms with a running time of $O(n \sqrt{ \lg \lg n})$ in the randomized case~\cite{han2002integer} and $O(n \lg \lg n)$ in the deterministic case~\cite{MR2121186} (both bounds assume that the word size and key size are within constant factors of each other). In external memory, no integer sorting algorithms faster than the comparison based $O((n/B)\lg_{M/B}(n/B))$ bound are known! Whether faster integer sorting algorithms exist was posed as an important open problem in the original paper by \citeboth{aggarwal1988input} that introduced the external memory model. Three decades later, we still do not know the answer to this question.

In this paper, we present tight conditional lower bounds for external memory integer sorting via a central conjecture by \citeboth{lili} in the area of \emph{network coding}. Our conditional lower bounds show that it is impossible to design integer sorting algorithms that outperform the optimal comparison based algorithms, thus settling the complexity of integer sorting under the conjecture by Li and Li.

\subsection{Network Coding}
The field of network coding studies the following communication problem over a network: Given a graph $G$ with capacity constraints on the edges and $k$ data streams, each with a designated source-sink pair of nodes $(s_i,t_i)$ in $G$, what is the maximum rate at which data can be transmitted concurrently between the source-sink pairs? A simple solution is to forward the data as indivisible packages, effectively reducing the problem to Multicommodity Flow (MCF). The key question in network coding, is whether one can achieve a higher rate by using coding/bit-tricks. This question is known to have a positive answer in directed graphs, where the rate increase may be as high as a factor $\Omega(|G|)$ (by sending XOR's of carefully chosen input bits), see e.g.~\cite{Adler:soda}. However the question remains wide open for undirected graphs where there are no known examples for which network coding can do anything better than the Multicommodity Flow rate. The lack of such examples resulted in the following central conjecture in network coding:
\begin{conj}[Undirected $k$-pairs Conjecture~\cite{lili}]
\label{conj:main}
The coding rate is equal to the Multicommodity Flow rate in undirected graphs.
\end{conj} 
Despite the centrality of this conjecture, it has so forth resisted all attempts at either proving or refuting it. 
Adler et al.~\cite{Adler:soda} made an exciting connection between the conjecture and lower bounds for algorithms. More concretely, they proved that if Conjecture~\ref{conj:main} is true, then one immediately obtains non-trivial lower bounds for all of the following:
\begin{itemize}
\item \emph{Oblivious} external memory algorithms.
\item \emph{Oblivious} word-RAM algorithms.
\item \emph{Oblivious} two-tape Turing machines.
\end{itemize}
In the above, \emph{oblivious} means that the memory access pattern of the algorithm (or tape moves of the Turing machine) is fixed and independent of the input data. Thus proving Conjecture~\ref{conj:main} would also give the first non-trivial lower bounds for all these classes of algorithms. One can view this connection in two ways: Either as exciting conditional lower bounds for (restricted) algorithms, or as a strong signal that proving Conjecture~\ref{conj:main} will be very difficult.

In this paper, we revisit these complexity theoretic implications of Conjecture~\ref{conj:main}. Our results show that the restriction to oblivious algorithms is unnecessary. In more detail, we show that Conjecture~\ref{conj:main} implies non-trivial (and in fact tight) lower bounds for external memory sorting of integers and for external memory matrix transpose algorithms. We also obtain tight lower bounds for word-RAM sorting algorithms when the word size is much larger than the key size, as well as tight lower bounds for transposing a $b \times b$ matrix on a word-RAM with word size $b$ bits. The striking thing is that our lower bounds hold without \emph{any} extra assumptions such as \emph{obliviousness}, \emph{indivisibility}, \emph{comparison-based} or the like. Thus proving Conjecture~\ref{conj:main} is as hard as proving super-linear algorithm lower bounds in the full generality word-RAM model, a barrier far beyond current lower bound techniques! Moreover, we show that the assumption from previous papers about algorithms being oblivious makes a huge difference for integer sorting: We prove an $\Omega(n \log n)$ lower bound for sorting $\Theta(\log n)$ bit integers using an oblivious word-RAM algorithm with word size $\Theta(\lg n)$ bits. This is in sharp contrast to the classic (non-oblivious) Radix Sort algorithm, which solves the problem in $O(n)$ time. Thus the previous restriction to oblivious algorithms may be very severe for some problems. 

\subsection{Lower Bounds for Sorting}
Our main result for external memory integer sorting is the following connection to Conjecture~\ref{conj:main}:

\begin{restatable}{theorem}{thmmain}
\label{thm:main}
Assuming Conjecture~\ref{conj:main}, any randomized algorithm for the external memory sorting problem with $w = \Omega(\lg n)$ bit integers, having error probability at most $1/3$, must make an expected 
$$
\Omega\Big(\min\Big\{n,\frac{n}{B} \cdot \lg_{2M/B} \frac{n}{B} \Big\}\Big)
$$
I/Os.
\end{restatable} 
Thus if we believe Conjecture~\ref{conj:main}, then even for randomized algorithms, there is no hope of exploiting integer input to improve over the simple external memory comparison based algorithms (when $B \geq \lg n$ such that the latter term in the lower bound is the min). 

Now observe that since our lower bound only counts I/Os, the lower bound immediately holds for word-RAM algorithms when the word size is some $b = \Omega(\lg n)$ by setting $m=O(b)$ and $B=b/w$ in the above lower bound (the CPU's internal state, i.e. registers, can hold only a constant number of words). Thus we get the following lower bound:
\begin{corollary}
\label{thm:RAM}
Assuming Conjecture~\ref{conj:main}, any randomized word-RAM algorithm for sorting $w = \Omega(\lg n)$ bit integers, having error probability at most $1/3$ and word size $b \geq w$ bits, must spend
$$
\Omega\Big(\min\Big\{n,\frac{n w}{b} \cdot \lg \frac{n w}{b} \Big\}\Big)
$$
time.
\end{corollary} 
We note that the a standard assumption in the word-RAM is a word size and key size of $b,w=\Theta(\lg n)$ bits. For that choice of parameters, our lower bound degenerates to the trivial $t = \Omega(n)$. This has to be the case, as Radix Sort gives a matching upper bound. Nonetheless, our lower bound shows that when the key size is much smaller than the word size, one cannot sort integers in linear time (recall linear is $O(n w/b)$ as this is the time to read/write the input/output).

Finally, we show that the obliviousness assumption made in the previous paper by Adler et al.~\cite{Adler:soda} allows one to prove very strong sorting lower bounds that even surpass the known (non-oblivious) Radix Sort upper bound:
\begin{theorem}
\label{thm:RAM}
Assuming Conjecture~\ref{conj:main}, any oblivious randomized word-RAM algorithm for sorting $\Theta(\lg n)$ bit integers, having error probability at most $1/3$ and word size $\Theta(\lg n)$, must spend
$
\Omega(n \lg n)
$
time.
\end{theorem}
Thus at least for the natural problem of integer sorting, being oblivious has a huge impact on the possible performance of algorithms. Our results are therefore not just an application of the previous technique to a new problem, but a great strengthening. Moreover, as we discuss in Section~\ref{sec:overview}, removing the obliviousness assumption requires new and deep ideas that result in significantly more challenging lower bound proofs.

\subsection{Lower Bounds for Matrix Transpose}
We also reprove an analog of the lower bounds by Adler et al.~\cite{Adler:soda} for the matrix transpose problem, this time without any assumptions of obliviousness. In the matrix transpose problem, the input is an $n \times n$ matrix $A$ with $w$-bit integer entries. The matrix is given in row-major order, meaning that each row of $A$ is stored in $n/B$ blocks of $B$ consecutive entries each. The goal is to compute $A^T$, i.e. output the column-major representation of $A$ which stores $n/B$ disk blocks for each column of $A$, each containing a consecutive range of $B$ entries from the column.

Based on Conjecture~\ref{conj:main}, \citeboth{Adler:soda} proved an $\Omega(B \lg B)$ lower bound on the number of I/Os needed for this problem when $n = B$, assuming that $M=2B$ and that the algorithm is oblivious. We strengthen the lower bound to the following:
\begin{theorem}
\label{thm:transposeIO}
Assuming Conjecture~\ref{conj:main}, any randomized algorithm for the external memory matrix transpose problem with $w$ bit integer entries, having error probability at most $1/3$, must make an expected 
$$
\Omega\left(\min\left\{\frac{n^2 \lg_{2M/B} B}{B} , \frac{n^2 w}{\lg(n^2 w)} \right\}\right)
$$
I/Os.
\end{theorem}
Notice that the first term in the min corresponds to their lower bound when $n=B$ (they have $M=2B$), but our new lower bound does not require the algorithm to be oblivious (and allows randomization). It is unclear whether the second term in the min is just an artifact of our proof, but we remark that it can only be the minimum when $B$ is very small.

Consider now the matrix transpose problem on the word-RAM with word size $b$ bits (and thus memory size $m=O(b)$). Given an $n \times n$ matrix $A$ with $w$-bit integer entries, the lower bound in Theorem~\ref{thm:transposeIO} implies (by setting $B = b/w$):
\begin{corollary}
Assuming Conjecture~\ref{conj:main}, any randomized word-RAM algorithm for computing the transpose of an $n \times n$ matrices with $w$-bit integer entries, having error probability at most $1/3$ and word size $b$ bits, must spend
$$
\Omega\left(\min\left\{\frac{n^2w \lg(b/w)}{b}, \frac{n^2w}{\lg(n^2 w)}  \right\}\right)
$$
time.
\end{corollary}
The above corollary in particular implies that for $b \times b$ bit matrices and word size $b$ bits, one needs $\Omega(b \lg b)$ time to transpose, whereas linear time would have been $O(b)$ as this is the time to read/write the input/output.

\subsection{Other Related Work}
\label{sec:prior}
Proving lower bounds for external memory algorithms and data structures without assumptions such as indivisibility and comparison based has been the focus of a number of recent papers. Quite surprisingly, \citeboth{iacono2012using} showed that for the \emph{dictionary} problem, one can indeed exploit integer input to beat the comparison based bounds. This is in sharp contrast to our new results for integer sorting. \citeboth{iacono2012using} complemented their upper bound by a matching unconditional lower bound proved in a version of the cell probe model of \citeboth{MR624736} adapted to the external memory setting. Their lower bound improved over previous work by \citeboth{zhang} and \citeboth{verbin}. In other very recent work, \citeboth{larsen:pqupper} showed how to exploit integer input to develop external memory priority queues with DecreaseKeys that outperform their comparison based counterparts. Their upper bound almost matches an unconditional lower bound by \citeboth{eenberg:pq} (also proved in the cell probe model).

Another active line of research has studied the benefits of network coding over the traditional routing-based solutions. Research on the network coding started by the work of \citeboth{ahlswede2000network}. They provided some examples which represent the benefits of network coding in the directed graphs. Later, the benefit of network coding in directed graph is considered in a sequence of works \cite{koetter2003algebraic,harvey2006capacity,harvey2004comparing}. It is known that there exists a family of directed graphs $G$ in which the gap between the coding rate and the MCF rate could be as large as $\Omega(|G|)$. 

Unlike directed graphs, it is conjectured by \citeboth{lili} (Conjecture~\ref{conj:main}) that in undirected graphs, the coding rate is equal to the MCF rate. Despite the persistent effort, this conjecture has been remained open for almost two decades. However, it has been shown that the conjecture holds in various classes of graphs. Specifically, it is known that the conjecture holds when the sparsity of the graph is equal to the MCF rate. Additionally, it is known that the coding rate cannot exceed the flow rate by more than a factor $\lg |G|$. This follows by relating the flow rate and coding rate to the sparsest cut. Other work by \cite{harvey2006capacity,jain2006capacity} showed the correctness of the conjecture for an infinite family of bipartite graphs. Also, in a recent paper by \citeboth{braverman2016network}, it is shown that if there is a graph where the coding rate exceeds the flow rate by a factor of $(1+\epsilon)$ for any constant $\eps>0$, then one can find an infinite family of graphs $\{G\}$ where the gap is a factor of $(\log |G|)^c$, where $0 < c < 1$ is a positive constant. It is also worth mentioning that a study by \citeboth{li2005achieving} gives empirical support for the conjecture and the paper~\cite{reductionapproach} uses exhaustive computer search to prove the conjecture for networks of up to six nodes.

\section{Preliminaries}
\label{sec:prelim}
We now give a formal definition of the $k$-pairs communication problem and the Multicommodity Flow problem.

\paragraph{$k$-pairs communication problem.} To keep the definition as simple as possible, we restrict ourselves to directed acyclic communication networks/graphs and we assume that the demand between every source-sink pair is the same. This will be sufficient for our proofs. For a more general definition, we refer the reader to~\cite{Adler:soda}.

The input to the $k$-pairs communication problem is a directed acyclic graph $G=(V,E)$ where each edge $e \in E$ has a capacity $c(e) \in \R^+$. There are $k$ sources $s_1,\dots,s_k \in V$ and $k$ sinks $t_1,\dots,t_k \in V$. Typically there is also a demand $d_i$ between each source-sink pair, but for simplicity we assume $d_i = 1$ for all pairs. This is again sufficient for our purposes.

Each source $s_i$ receives a message $A_i$ from a predefined set of messages $A(i)$. It will be convenient to think of this message as arriving on an in-edge. Hence we add an extra node $S_i$ for each source, which has a single out-edge to $s_i$. The edge has infinite capacity.

A network coding solution specifies for each edge $e \in E$ an alphabet $\Gamma(e)$ representing the set of possible messages that can be sent along the edge. For a node $v \in V$, define $\In(u)$ as the set of in-edges at $u$. A network coding solution also specifies, for each edge $e=(u,v) \in E$, a function $f_e : \prod_{e' \in \In(u)} \Gamma(e') \to \Gamma(e)$ which determines the message to be sent along the edge $e$ as a function of all incoming messages at node $u$. Finally, a network coding solution specifies for each sink $t_i$ a decoding function $\sigma_i : \prod_{e \in \In(t_i)} \Gamma(e) \to M(i)$. The network coding solution is correct if, for all inputs $A_1,\dots,A_k \in \prod_i A(i)$, it holds that $\sigma_i$ applied to the incoming messages at $t_i$ equals $A_i$, i.e. each source must receive the intended message.

In an execution of a network coding solution, each of the extra nodes $S_i$ starts by transmitting the message $A_i$ to $s_i$ along the edge $(S_i,s_i)$. Then, whenever a node $u$ has received a message $a_e$ along all incoming edges $e=(v,u)$, it evaluates $f_{e'}(\prod_{e \in \In(u)} a_e)$ on all out-edges and forwards the message along the edge $e'$.

Following Adler et al.~\cite{Adler:soda} (and simplified a bit), we define the \emph{rate} of a network coding solution as follows: Let each source receive a uniform random and independently chosen message $A_i$ from $A(i)$. For each edge $e$, let $A_e$ denote the random variable giving the message sent on the edge $e$ when executing the network coding solution with the given inputs. The network coding solution achieves rate $r$ if:
\begin{itemize}
\item $H(A_i) \geq r d_i = r$ for all $i$.
\item For each edge $e \in E$, we have $H(A_e) \leq c(e)$.
\end{itemize}
Here $H(\cdot)$ denotes binary Shannon entropy. The intuition is that the rate is $r$, if the solution can handle upscaling the entropy of all messages by a factor $r$ compared to the demands.

\paragraph{Multicommodity Flow.}
A multicommodity flow problem in an undirected graph $G=(V,E)$ is specified by a set of $k$ source-sink pairs $(s_i,t_i)$ of nodes in $G$. We say that $s_i$ is the source of commodity $i$ and $t_i$ is the sink of commodity $i$. Each edge $e \in E$ has an associated capacity $c(e) \in \R^+$. In addition, there is a demand $d_i$ between every source-sink pair. For simplicity, we assume $d_i = 1$ for all $i$ as this is sufficient for our needs.

A (fractional) solution to the multicommodity flow problem specifies for each pair of nodes $(u,v)$ and commodity $i$, a flow $f_i(u,v) \in [0,1]$. Intuitively $f_i(u,v)$ specifies how much of commodity $i$ that is to be sent from $u$ to $v$. The flow satisfies \emph{flow conservation}, meaning that:
\begin{itemize}
\item For all nodes $u$ that is not a source or sink, we have $\sum_{w \in V} f_i(u,w) - \sum_{w \in V} f_i(w,u) = 0$.
\item For all sources $s_i$, we have $\sum_{w \in V} f_i(s_i,w) - \sum_{w \in V}f_i(w,s_i) = 1$.
\item For all sinks we have $\sum_{w \in V} f_i(w,t_i) - \sum_{w \in V} f_i(t_i,w) = 1$.
\end{itemize}
The flow also satisfies that for any pair of nodes $(u,v)$ and commodity $i$, there is only flow in one direction, i.e. either $f_i(u,v)=0$ or $f_i(v,u)=0$. Furthermore, if $(u,v)$ is not an edge in $E$, then $f_i(u,v) = f_i(v,u)=0$. A solution to the multicommodity flow problem achieves a rate of $r$ if:
\begin{itemize}
\item For all edges $e=(u,v) \in E$, we have $r \cdot \sum_i d_i (f_i(u,v) + f_i(v,u)) = r \cdot \sum_i (f_i(u,v) + f_i(v,u)) \leq c(e)$.
\end{itemize}
Intuitively, the rate is $r$ if we can upscale the demands by a factor $r$ without violating the capacity constraints.

\paragraph{The Undirected $k$-pairs Conjecture.} Conjecture~\ref{conj:main} implies the following for our setting: Given an input to the $k$-pairs communication problem, specified by a directed acyclic graph $G$ with edge capacities and a set of $k$ source-sink pairs with a demand of $1$ for every pair, let $r$ be the best achievable network coding rate for $G$. Similarly, let $G'$ denote the undirected graph resulting from making each directed edge in $G$ undirected (and keeping the capacities, source-sink pairs and a demand of $1$ between every pair). Let $r'$ be the best achievable flow rate in $G'$. Conjecture~\ref{conj:main} implies that $r \leq r'$.

Having defined coding rate and flow rate formally, we also mention that the result of \citeboth{braverman2016network} implies that if there exists a graph $G$ where the network coding rate $r$, and the flow rate $r'$ in the corresponding undirected graph $G'$, satisfies $r \geq (1+\eps)r'$ for a constant $\eps>0$, then there exists an infinite family of graphs $\{G^*\}$ for which the corresponding gap is at least $(\lg |G^*|)^c$ for a constant $c>0$. So far, all evidence suggest that no such gap exists, as formalized in Conjecture~\ref{conj:main}.

\section{Proof Overview}
\label{sec:overview}
In this section, we give an overview of the main ideas in our proof and explain the barriers we overcome in order to remove the assumption of obliviousness. To prove our lower bound for external memory sorting, we focus on the easier problem of permuting. In the permutation problem, we are given an array $A$ of $n$ entries. The $i$'th entry of $A$ stores a $w$-bit data item $d_i$ and a destination $\pi(i)$. The destinations $\pi(i)$ form a permutation $\pi$ of $\{1,\dots,n\}$. The goal is to produce the output array $C$ where $d_i$ is stored in entry $C[\pi(i)]$. The arrays $A$ and $C$ are both stored in disk blocks, such that each disk block of $A$ stores $(\lg n + w)/b$ entries, and each disk block of $C$ stores $w/b$ entries (the maximum number of entries that can be packed in a block). A sorting algorithm that can sort $(\lg n + w)$ bit integer keys can be used to solve the permutation problem by replacing each entry $(\pi(i),d_i)$ with the integer $\pi(i) \cdot 2^w + d_i$ (in the addition, we think of $d_i$ as an integer in $[2^w]$). Thus it suffices to prove lower bounds for permuting.

Consider now an algorithm $\Alg$ for permuting, and assume for simplicity that it is deterministic and always correct. As in the previous work by Adler et al.~\cite{Adler:soda}, we define a graph $G(A)$ that captures the memory accesses of $\Alg$ on an input array $A$. The graph $G$ has a node for every block in the input array, a node for every block in the output and a node for every intermediate block written/read by $\Alg$. We call these block nodes. Moreover, the graph has a memory node that represent the memory state of $\Alg$. The idea is that whenever $\Alg$ reads a block into memory, then we add a directed edge from the corresponding block node to the memory node. When $\Alg$ writes to a block, we create a new node (that replaces the previous version of the block) and add a directed edge from the memory node to the new node. The algorithm $\Alg$ can now be used to send messages between input and output block nodes as follows: Given messages $X_1,\dots,X_n$ of $w$ bits each and an intended output block node (storing $C[\pi(i)]$) for each message $i$, we can transmit the message $X_i$ from the input block node representing the array entry $A[i]$ to the output block node representing the array entry $C[\pi(i)]$ simply by simulating the algorithm $\Alg$: Each block node of the network always forward any incoming message to the memory node along its outgoing edge. The memory node thus receives the contents of all blocks that it ever reads. It can therefore simulate $\Alg$. Whenever it performs a write operation, it sends the contents along the edge to the designated block node. By the correctness of $\Alg$, this results in every output block node knowing the contents of all array entries $C[\pi(i)]$ that should be stored in that output block. Examining this simulation, we see that we need a capacity of $b$ bits on all edges for the simulation to satisfy capacity constraints. Moreover, by the definition of network coding rate (Section~\ref{sec:prelim}), we see that the coding rate is $w$ bits.

The idea is that we want to use Conjecture~\ref{conj:main} to argue that the graph $G$ must be large (i.e. there must be many I/Os). To do so, we would like to argue that if we undirect $G$, then there is a permutation $\pi$ such that for many pairs $A[i]$ and $C[\pi(i)]$, there are no short paths between the block nodes storing $A[i]$ and $C[\pi(i)]$. If we could argue that for $n/2$ pairs $(A[i],C[\pi(i)])$, there must be a distance of at least $\ell$ steps in the undirected version of $G$, then to achieve a flow rate of $w$, it must be the case that the sum of capacities in $G$ is at least $\ell wn/2$. But each I/O adds only $2b$ bits of capacity to $G$. Thus if $\Alg$ makes $t$ I/Os, then it must be the case that $tb = \Omega(\ell wn) \Rightarrow t = \Omega((nw/b) \cdot \ell) = \Omega((n/B) \cdot \ell)$.

Unfortunately, we cannot argue that there must be a long path between many pairs in the graph $G$ we defined above. The problem is that the memory node is connected to all block nodes and thus the distance is never more than $2$. To fix this, we change the definition of $G$ slightly: After every $m/b$ I/Os, we deactivate the memory node and create a new memory node to replace it. Further I/Os insert edges to and from this new memory node. In order for the new memory node to continue the simulation of $\Alg$, the new memory node needs to know the memory state of $\Alg$. Hence we insert a directed edge from the old deactivated memory node to the new memory node. The edge has capacity $m$ bits. Thus in the simulation, when the current memory node has performed $m/b$ I/Os, it forwards the memory state of $\Alg$ to the next memory node who continues the simulation. The $m/b$ I/Os between the creation of new memory nodes has been chosen such that the amortized increase in capacity due to an I/O remains $O(b)$. 

We have now obtained a graph $G$ where the degrees of all nodes are bounded by $2m/b$. Thus for every node $G$, there are at most $(2m/b)^{\ell}$ nodes within a distance of $\ell$. Thus intuitively, a random permutation $\pi$ should have the property that for most pairs $(A[i],C[\pi(i)])$, there will be a distance of $\ell=\Omega(\lg_{2m/b} n/B)$ between the corresponding block nodes. This gives the desired lower bound of $t = \Omega((n/B) \cdot \ell) =  \Omega((n/B) \cdot \lg_{2m/b} n/B)$.

If we had assumed that the algorithm $\Alg$ was oblivious as in previous work, we would actually be done by now. This is because, under the obliviousness assumption, the graph $G$ will be the same \emph{for all} input arrays. Thus one can indeed find the desired permutation $\pi$ where there is a large distance between most pairs $(A[i],C[\pi(i)])$. Moreover, all inputs corresponding to that permutation $\pi$ and data bit strings $d_1,\dots,d_n$ can be simulated correctly using $\Alg$ and the graph $G$. Hence one immediately obtains a network coding solution. However, when $\Alg$ is not constrained to be oblivious, there can be a large number of distinct graphs $G$ resulting from the execution of $\Alg$.

To overcome this barrier, we first argue that even though there can be many distinct graphs, the number of such graphs is still bounded by roughly $(nw/b + t)^t$ (each I/O chooses a block to either read or write and there are $t$ I/Os). This means that for $t = o(n)$, one can still find a graph $G$ that is the result of running $\Alg$ on many different input arrays $A$. We can then argue that amongst all those inputs $A$, there are many that all correspond to the same permutation $\pi$, and that permutation $\pi$ has the property from before that, for most pairs $(A[i],C[\pi(i)])$, there will be a distance of $\ell=\Omega(\lg_{2m/b} n/B)$ between the corresponding block nodes. Thus we would like to fix such a permutation and use $\Alg$ to obtain a network coding solution. The problem is that we can only argue that there are \emph{many} data bit strings $d_1,\dots,d_n$ that together with $\pi$ result in an array $A$ for which $\Alg$ uses the graph $G$. Thus we can only correctly transmit a large collection of messages, not all messages. Let us call this collection $\Fam \subseteq \{\{0,1\}^w\}^n$ and let us assume $|\Fam| \geq 2^{nw-o(nw)}$. Intuitively, if we draw a uniform random input from $\Fam$, then we should have a network coding solution with a rate of $w-o(w)$. The problem is, that the definition of network coding requires the inputs to the nodes to be independent. Thus we cannot immediately say that we have a network coding solution with rate $w - o(w)$ by solving a uniform random input from $\Fam$. To remedy this, we instead take the following approach: We let each data bit string $d_i$ be a uniform random and independently chosen $w$-bit string. Thus if we can solve the network coding problem with these inputs, then we indeed have a network coding solution. We would now like to find an efficient way of translating the bit strings $d_1,\dots,d_n$ to new bit strings $d'_1,\dots,d'_n$ with $d'_1,\dots,d'_n \in \Fam$. The translation should be such that each input block node can locally compute the $d'_i$, and the output block nodes should be able to revert the transformation, i.e. compute from $d'_i$ the original bit string $d_i$. To achieve this, we need to modify $G$ a bit. Our idea is to introduce a coordinator node that can send short descriptions of the mappings between the $d_i$s and $d'_i$s. We accomplish this via the following lemma that we prove in Section~\ref{sec:coordinate}:
\begin{lemma}
\label{lem:coordinate}
Consider a communication game with a coordinator $u$, a set $\Fam \subseteq \{0,1\}^{nw}$ and $n$ players. Assume $|\Fam| \geq 2^{nw-r}$ for some $r$. The coordinator receives as input $n$ uniform random bit strings $X_i$ of $w$ bits each, chosen independently of the other $X_j$. The coordinator then sends a prefix-free message $R_i$ to the $i$'th player for each $i$. From the message $R_i$ alone (i.e. without knowing $X_i$), the $i$'th player can then compute a vector $\tau_i \in \{0,1\}^w$ with the property that the concatenation $q := (\tau_1 \oplus X_1) \circ (\tau_2 \oplus X_2) \circ \cdots \circ (\tau_n \oplus X_n)$ satisfies $q \in \Fam$, where $\oplus$ denotes bit wise XOR. There exists such a protocol where
$$
\sum_i \E[|R_i|] = O\left(n+ \sqrt{nwr}\lg(nw/r)\right).
$$
In particular, if $r=o(nw)$ and $w=\omega(1)$ then the communication satisfies:
$$
\sum_i \E[|R_i|] = o(nw).
$$
\end{lemma}
We use the lemma as follows: We create a coordinator node $u$ that is connected to all input block nodes and all output block nodes. In a simulation of $\Alg$, the input block nodes start by transmitting their inputs to the coordinator node $u$. The coordinator then computes the messages in the lemma and sends $R_i$ back to the input block node storing $A[i]$ as well as to the output block node storing the array entry $C[\pi(i)]$. The input block nodes can now compute $d'_i = \tau_i \oplus d_i$ to obtain an input $d'_1,\dots,d'_n \in \Fam$. We can then run the algorithm $\Alg$ since this is an input that actually results in the graph $G$. Finally, the output block nodes can revert the mapping by computing $d_i = \tau_i \oplus d'_i$. Thus what the lemma achieves, is an efficient way of locally modifying the inputs of the nodes, so as to obtain an input for which the algorithm $\Alg$ works. We find this contribution very novel and suspect it might have applications in other lower bound proofs.

The introduction of the node $u$ of course allows some flow to traverse paths not in the original graph $G$. Thus we have to be careful with how we set the capacities on the edges to and from $u$. We notice that edges from the input nodes to $u$ need only a capacity of $w$ bits per array entry (they send the inputs), and edges out of $u$ need $\E[|R_i|]$ capacity for an input $d_i$ (one such edge to the input block node for array entry $A[i]$ and one such edge to the output block node for array entry $C[\pi(i)]$). The crucial observation is that any flow using the node $u$ as an intermediate node, must traverse at least two edges incident to $u$. Hence only $(nw + 2 \sum_i \E[|R_i|])/2$ flow can traverse such paths. If $|\Fam| \geq 2^{nw-o(nw)}$ then Lemma~\ref{lem:coordinate} says that this is no more than $nw/2 + o(nw)$ flow. There therefore remains $nw/2 - o(nw)$ flow that has to traverse the original length $\ell= \Omega(\lg_{2m/b} n/B)$ paths and the lower bound follows. 

One may observe that our proof uses the fact that the network coding rate is at most the flow rate in a strong sense. Indeed, the introduction of the node $u$ allows a constant fraction of the flow to potentially use a constant length path. Thus it is crucial that the network coding rate $r$ and flow rate $r'$ is conjectured to satisfy $r \leq r'$ and not e.g. $r \leq 3r'$. Indeed we can only argue that a too-good-to-be-true permutation algorithm yields a graph in which $r \geq ar'$ for some constant $a > 1$. However, as pointed out in Section~\ref{sec:prior}, \citeboth{braverman2016network} recently proved that if there is a graph where $r \geq (1+\eps)r'$ for a constant $\eps>0$, then there is an infinite family of graphs $\{G'\}$ where the gap is $\Omega((\lg |G'|)^c)$ for a constant $c>0$. Thus a too-good-to-be-true permutation algorithm will indeed give a strong counter example to Conjecture~\ref{conj:main}.

Our proof of Lemma~\ref{lem:coordinate} is highly non-trivial and is based on the elegant proof of the $\sqrt{IC}$ bound by \citeboth{Barak:compress} for compressing interactive communication under non-product distributions. Our main idea is to argue that for a uniform random bit string in $\{0,1\}^{nw}$ (corresponding to the concatenation $X = X_1 \circ \cdots \circ X_n$ of the $X_i$'s in the lemma), it must be the case that the expected Hamming distance to the nearest bit string $Y$ in $\Fam$ is $O(\sqrt{nwr})$. The coordinator thus finds $Y$ and transmits the XOR $X \oplus Y$ to the players. The XOR is sparse and thus the message can be made short by specifying only the non-zero entries. Proving that the expected distance to the nearest vector is $O(\sqrt{nwr})$ is the main technical difficulty and is the part that uses ideas from protocol compression.
\section{External Memory Lower Bounds}
As mentioned in the proof overview in Section~\ref{sec:overview}, we prove our lower bound for external memory sorting via a lower bound for the easier problem of permuting:
An input to the permutation problem is specified by a permutation $\pi$ of $\{1, 2, \dots, n\}$ as well as $n$ bit strings $d_1,\dots,d_n \in \{0,1\}^w$. We assume $w \geq \lg n$ such that all bit strings may be distinct. The input is given in the form of an array $A$ where the $i$'th entry $A[i]$ stores the tuple $(\pi(i), d_i)$. We assume the input is given in the following natural way: Each $\pi(i)$ is encoded as a $\lceil\log n \rceil$-bit integer and the $d_i$'s are given as they are - using $w$ bits for each.

The array $A$ is presented to an external memory algorithm as a sequence of blocks, where each block contains $\lfloor b/(w + \log n) \rfloor$ consecutive entries of $A$ (the blocks have $b = Bw$ bits). For simplicity, we henceforth assume $(w + \log n)$ divides $b$.

The algorithm is also given an initially empty output array $C$. The array $C$ is represented as a sequence of $n$ words of $w$ bits each, and these are packed into blocks containing $b/w$ words each. The goal is to store $d_{\pi^{-1}(i)}$ in $C[i]$. That is, the goal is to \emph{copy} the bit string $d_i$ from $A[i]$ to $C[\pi(i)]$. We say that an Algorithm $\mathcal{A}$ has an error of $\varepsilon$ for the permutation problem, if for every input to the problem, it produces the correct output with the probability at least $1-\varepsilon$.

The best known upper bounds for the permutation problem work also under the \emph{indivisibility} assumption. These algorithms solve the permutation problem in 
$$
O\left(\min\left\{n, \frac{nw}{b} \cdot \log_{m/b}(nw/b)\right\}\right) = O\left(\min\left\{n, \frac{n}{B} \cdot \log_{M/B}(n/B)\right\}\right) 
$$ I/Os \cite{aggarwal1988input}. Moreover, this can easily be shown to be optimal under the indivisibility assumption by using a counting argument~\cite{aggarwal1988input}. The $n$ bound is the bound obtained by running the naive ``internal memory'' algorithm that simply puts each element into its correct position one at a time. The other term is equivalent to the optimal comparison-based sorting bound (one thinks of $d_i$ as an integer in $[2^w]$ and concatenates $\pi(i) \circ d_i = \pi(i) \cdot 2^w + d_i$ and sorts the sequence). Thus any sorting algorithm that handles $(\lg n + w)$-bit keys immediately yields a permutation algorithm with the same number of I/Os. We thus prove lower bounds for the permutation problem and immediately obtain the sorting lower bounds as corollaries.

We thus set out to use Conjecture~\ref{conj:main} to provide a lower bound for the permutation problem in the external memory model. Throughout the proof, we assume that $nw/b = n/B$ is at least some large constant. This is safe to assume, as otherwise we only claim a trivial lower bound of $\Omega(1)$.

Let $\Alg$ be a randomized external memory algorithm for the permutation problem on $n$ integers of $w$ bits each. Assume $\Alg$ has error probability at most $1/3$ and let $b$ denote the disk block size in number of bits. Let $m$ denote the memory size measured in number of bits. Finally, let $t$ denote the expected number of I/Os made by $\Alg$ (on the worst input).

\paragraph{I/O-Graphs.}
For an input array $A$ representing a permutation $\pi$ and bit strings $d_1,\dots,d_n$, and an output array $C$, define the (random) I/O-graph $G$ of $\Alg$ as follows: Initialize $G$ to have one node per disk block in $A$ and one node per disk block in $C$. Also add one node to $G$ representing the initial memory of $\Alg$. We think of the nodes representing the disk blocks of $A$ and $C$ as \emph{block nodes} and the node representing the memory as a \emph{memory node} (see Figure \ref{fig:IO-a}). We will add more nodes and edges to $G$ by observing the execution of $\Alg$ on $A$. To simplify the description, we will call nodes of $G$ either \emph{dead} or \emph{live}. We will always have at most one live memory node. Initially all nodes are live. Moreover, we label the block nodes by consecutive integers starting at $0$. Thus the block nodes in the initial graph are labeled $1,2,\dots,n(w+\lg n)/b + nw/b$.

Now run the algorithm $\Alg$ on $A$. Whenever it makes an I/O, do as follows: If this is the first time the block is being accessed and it is not part of the input or output (a write operation to an untouched block), create a new live block node in $G$ and add a directed edge from the current live memory node to the new block node (see Figure \ref{fig:IO-e}). Label the new node by the next unused integer label. Otherwise, let $v$ be the live node in $G$ corresponding to the last time the disk block was accessed. 
We add a directed edge from $v$ to the live memory node, mark $v$ as dead, create a new live block node $v'$ and add a directed edge from the live memory node to $v'$. We give the new node the same label as $v$ (Figure \ref{fig:IO-b} and Figure \ref{fig:IO-c}). Finally, once for every $m/b$ I/Os, we mark the memory node as dead, create a new live memory node and add an directed edge from the old memory node to the new live memory node (Figure \ref{fig:IO-d}).

To better understand the definition of $G$, observe that all the nodes with the same label represent the different versions of a disk block that existed throughout the execution of the algorithm. Moreover, there is always exactly one live node with any fixed label, representing the current version of the disk block. Also observe that at the end of the execution, there must be a live disk block node in $G$ representing each of the output blocks in $C$, and these have the same labels as the original nodes representing the empty disk blocks of $C$ before the execution of $\Alg$.

\begin{figure}
\centering
\begin{subfigure} [b] {0.49\textwidth}
  \includegraphics[width=\textwidth]{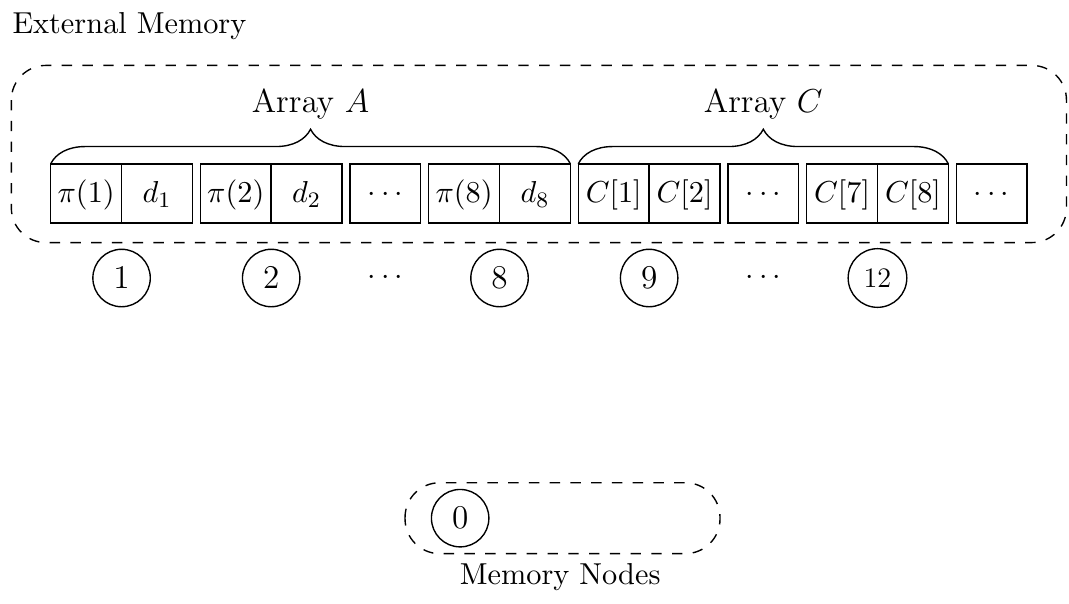}
  \vspace{-0.5cm}
  \caption{}
  \label{fig:IO-a}
 \end{subfigure}
 \begin{subfigure} [b] {0.49\textwidth}
  \includegraphics[width=\textwidth]{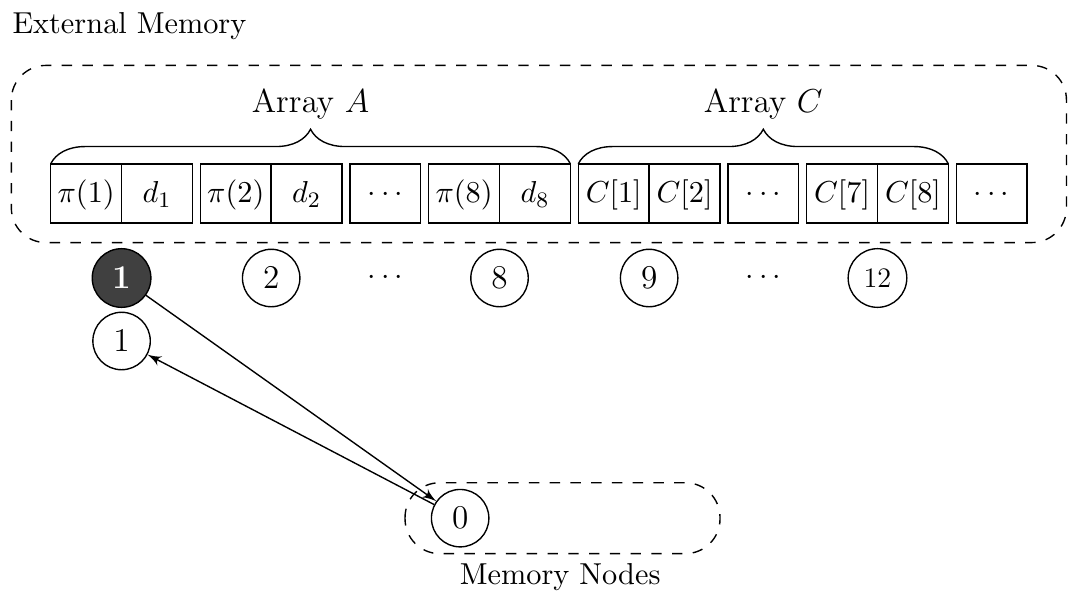}
  \vspace{-0.5cm}
  \caption{}
  \label{fig:IO-b}
 \end{subfigure}
  \begin{subfigure} [b] {0.49\textwidth}
  \includegraphics[width=\textwidth]{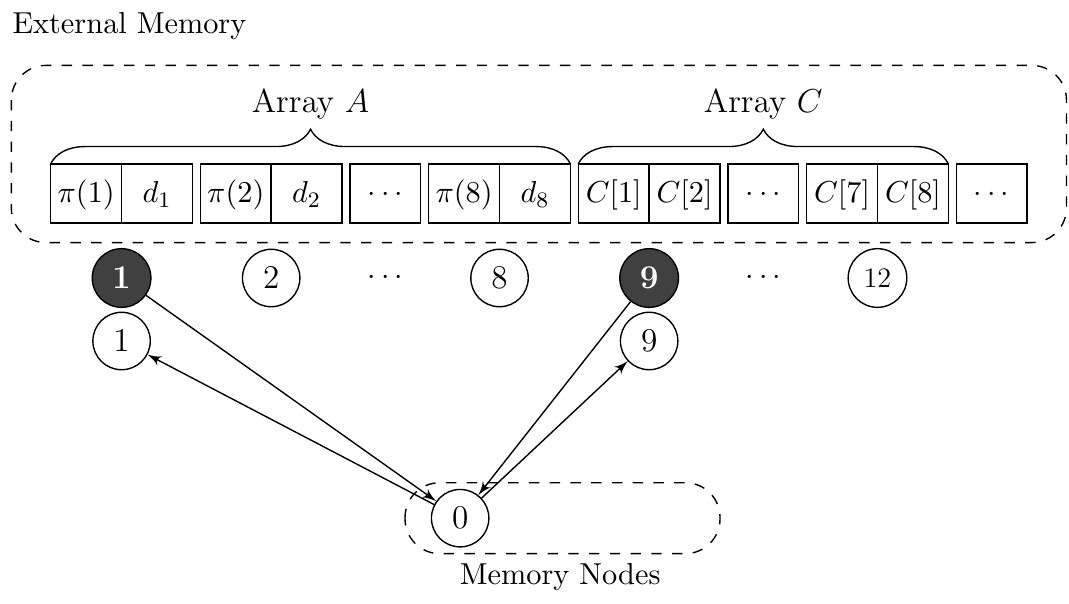}
  \vspace{-0.5cm}
  \caption{}
  \label{fig:IO-c}
 \end{subfigure}
  \begin{subfigure} [b] {0.49\textwidth}
  \includegraphics[width=\textwidth]{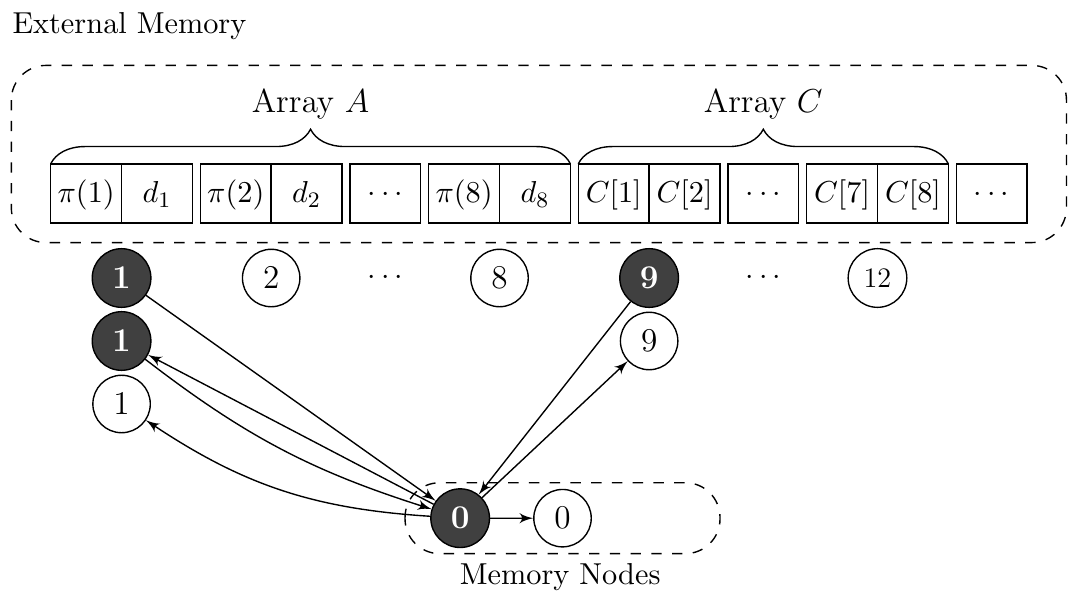}
  \vspace{-0.5cm}
  \caption{}
  \label{fig:IO-d}
 \end{subfigure}
  \begin{subfigure} [b] {0.6\textwidth}
  \includegraphics[width=\textwidth]{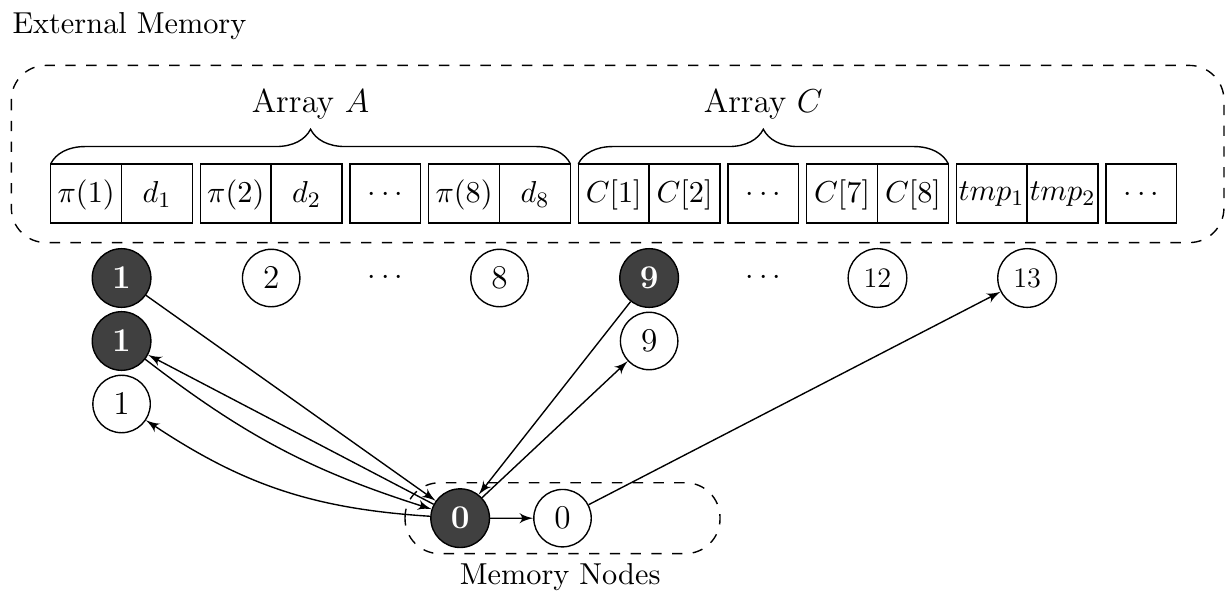}
  \vspace{-0.5cm}
  \caption{}
  \label{fig:IO-e}
 \end{subfigure}
 \caption{I/O-graph for an array $A$ consisting of $3$-bit strings $d_1, \cdots, d_8$ . In this example, each disk block contains two words of $w=3$ bits, i.e., $B=2$ (and $b=Bw=6$). Also, the main memory holds $M=6$ words ($m=18$). Figure (a) shows the initial I/O graph. For each disk block, we have initially one block node which is illustrated underneath them. Black nodes are dead, and white nodes are live. Figure (b) shows the updated I/O-graph after making an I/O to access the first disk block. Figure (c) is the I/O graph after accessing the block containing $C[1]$ and $C[2]$. Figure (d) shows the graph after making another I/O on the first disk block. Also, we create a new memory node after every $m/b=M/B=3$ I/Os and mark the old memory node as dead. Figure (e) shows the updated graph after accessing some block other than the input or output.
 }
 \label{fig:IO}
\end{figure}

\paragraph{Fixing the Randomness of $\Alg$.}
Consider the execution of $\Alg$ on an input $A$ representing a uniform random permutation $\pi$ as well as independent and uniform random bit strings $d_1,\dots,d_n \in \{0,1\}^w$. Since $\Alg$ makes an expected $t$ I/Os, it follows by Markov's inequality that $\Alg$ makes more than $6t$ I/Os with probability less than $1/6$. If we simply abort in such cases, we obtain an algorithm with worst case $O(t)$ I/Os and error probability at most $1/3 + 1/6 = 1/2$. Now fix the random choices of $\Alg$ to obtain a deterministic algorithm $\Alg^*$ with error probability $1/2$ over the random choice of $\pi$ and $d_1,\dots,d_n$. $\Alg^*$ makes $t^* = 6t$ I/Os in the worst case. Observe that for $\Alg^*$, we get a fixed I/O graph $G(A)$ for every input array $A$ since $\Alg^*$ is deterministic.

\paragraph{Finding a Popular I/O-Graph.}
We now find an I/O-graph $G$ which is the result of running $\Alg^*$ on a large number of different inputs. For notational convenience, let $t$ denote the worst case number of I/Os made by $\Alg^*$ (instead of using $t^*$ or $6t$). Observe that the total number of different I/O-graphs one can obtain as the result of running $\Alg^*$ is small:
\begin{lemma}
\label{lem:fewgraphs}
There are no more than
$$
(t+n(w+\lg n)/b + nw/b + 1)^{t+1}
$$ I/O graphs that may result from the execution of $\Alg^*$.
\end{lemma}
\begin{proof}
There are at most $t+1$ choices for the number of I/Os performed by $\Alg^*$ on an input ($0$ through $t$). Moreover, for each I/O performed, either one of at most $t+n(w+\lg n)/b + nw/b$ live disk blocks is read, or a new untouched disk block is created, resulting in at most $t+n(w+\lg n)/b  + nw/b+ 1$ different possible changes to $G$ for each I/O. Thus the total number of possible I/O-graphs is at most $(t+1)(t+n(w+\lg n)/b  + nw/b+ 1)^t \leq (t+n(w+\lg n)/b + nw/b + 1)^{t+1}$. 
\end{proof}

This means that we can find an I/O-graph, which correspond to the execution of $\Alg^*$ on many different inputs, and moreover, we can even assume that $\Alg^*$ is correct on many such inputs:
\begin{lemma}
\label{lem:popular}
There exists a set $\Gamma$ containing at least 
$
(n!2^{nw})/(2 (t+n(w+\lg n)/b  + nw/b+ 1)^{t+1})
$
different input arrays $A$, such that $\Alg^*$ is correct on all inputs $A \in \Gamma$ and the I/O-graph is the same for all $A \in \Gamma$.
\end{lemma}
\begin{proof}
There are $n!2^{wn}$ different input arrays $A$ and $\Alg^*$ had error probability at most $1/2$ over a uniform random choice of such an input array. By Lemma~\ref{lem:fewgraphs}, we get that there must be some I/O-graph shared by at least $(n!2^{nw}/2)/ (t+n(w+\lg n)/b + nw/b + 1)^{t+1}$ of the inputs that $\Alg^*$ is correct on.
\end{proof}

\paragraph{Data Must Travel Far.}
The key idea in our lower bound proof, is to argue that there is a permutation for which most data bit strings $d_i$ are very far away from output entry $C[\pi(i)]$ in the corresponding I/O-graph. This would require the data to ``travel'' far. By Conjecture~\ref{conj:main}, this is impossible unless the I/O-graph is large. Thus we start by arguing that there is a fixed permutation where data has to travel far on the average, and where it also holds that there are many different data values that can be sent using the same I/O-graph. To make this formal, let $\dist(\pi,i,G)$ denote the distance between the block node in $G$ representing the input block storing $A[i]$ (the initial node, before any I/Os were performed) and the node in $G$ representing the output block storing $C[\pi(i)]$ in the \emph{undirected} version of $G$ (undirect all edges).

We prove the following:
\begin{lemma}
\label{lem:travel}
If $(t + n(w+\lg n)/b + nw/b+1)^{t+1} \leq (nw/b)^{(1/30)n}$, then there exists a permutation $\pi$, a collection of values $\Fam \subseteq \{\{0,1\}^w\}^n$ and an I/O-graph $G$ such that the following holds:
\begin{enumerate}
\item For all $(d_1,\dots,d_n) \in \Fam$ it holds that the algorithm $\Alg^*$ executed on the input array $A$ corresponding to inputs $\pi$ and $d_1,\dots,d_n$ results in the I/O-graph $G$ and $\Alg^*$ is correct on $A$.
\item $|\Fam| \geq \frac{2^{nw}}{4 (t+n(w+\lg n)/b  + nw/b+ 1)^{t+1}}$.
\item There are at least $(4/5)n$ indices $i \in \{1,\dots,n\}$ for which $\dist(\pi,i,G) \geq (1/2)\log_{2m/b} (nw/b)$.
\end{enumerate}
\end{lemma}

\begin{proof}
We start by using Lemma~\ref{lem:popular} to obtain a set $\Gamma$ and an I/O-graph $G$ such that $$
|\Gamma| \geq \frac{n!2^{nw}}{2 (t+n(w+\lg n)/b  + nw/b+ 1)^{t+1}}.
$$
$\Gamma$ has the property that $\Alg^*$ is correct on all arrays $A \in \Gamma$ and $G$ is the I/O-graph corresponding to the execution of $\Alg^*$ on $A$ for every $A \in \Gamma$. For each permutation $\pi$, define $\Gamma_\pi$ to be the subset of arrays in $\Gamma$ for which the corresponding permutation is $\pi$. We have $|\Gamma_\pi| \leq 2^{nw}$ for all $\pi$. We now argue that there must be many $\Gamma_\pi$'s that are large: Let $k$ be the number of permutations $\pi$ such that
$$
|\Gamma_\pi| \geq \frac{2^{nw}}{4 (t+n(w+\lg n)/b  + nw/b+ 1)^{t+1}}.
$$
We must have 
\begin{eqnarray*}
k2^{nw} + (n!-k) \cdot \frac{2^{nw}}{4 (t+n(w+\lg n)/b  + nw/b+ 1)^{t+1}} &\geq& |\Gamma| \Rightarrow \\
k\left(2^{nw} -\frac{2^{nw}}{4 (t+n(w+\lg n)/b  + nw/b+ 1)^{t+1}}\right) &\geq& \frac{n!2^{nw}}{4 (t+n(w+\lg n)/b  + nw/b+ 1)^{t+1}} \Rightarrow \\
k2^{nw} &\geq&  \frac{n!2^{nw}}{4 (t+n(w+\lg n)/b  + nw/b+ 1)^{t+1}} \Rightarrow \\
k &\geq& \frac{n!}{4 (t+n(w+\lg n)/b  + nw/b+ 1)^{t+1}}
\end{eqnarray*}
If we assume $(t + n(w+\lg n)/b + nw/b+1)^{t+1} \leq (nw/b)^{(1/30)n}$ as in the statement of the lemma, then we get
$$
k \geq \frac{n!}{4(nw/b)^{(1/30)n}}.
$$
We have now argued that there are many permutations $\pi$, all having many arrays $A$ corresponding to $\pi$ with some data bit strings $d_1,\dots,d_n$, and where the the I/O-graph of $\Alg^*$ on $A$ is $G$. We will use this to conclude that for at least one of those permutations, it must be the case that $\dist(\pi,i,G)$ is large for many $i$. For this, observe that the number of distinct permutations $\pi$ for which there are less than $(4/5)n$ indices $i \in \{1,\dots,n\}$ with $\dist(\pi,i,G) \geq (1/2)\lg_{2m/b} (nw/b)$ is bounded by:
$$
\binom{n}{n/5} n^{(4/5)n} (2m/b)^{(n/5) \cdot (1/2)\lg_{2m/b} (nw/b)} (b/w)^{n/5}.
$$
To see this, observe that any such permutation $\pi$ can be uniquely specified by first specifying a set $I$ of $n/5$ indices $i$ with $\dist(\pi,i,G) < (1/2)\lg_{2m/b} (nw/b)$. There are $\binom{n}{n/5}$ possible choices for $I$. Then, for all indices $i$ with $i \notin I$, there are at most $n$ choices for $\pi(i)$. Finally, for indices $i \in I$ we argue as follows: Every node in $G$ has degree at most $2m/b$ by construction. Hence any node has at most $(2m/b)^\ell$ nodes within distance $\ell$ in the undirected version of $G$. Since $\dist(\pi,i,G) < (1/2)\lg_{2m/b} (nw/b)$ for all $i \in I$, it must be the case that the output node containing $C[\pi(i)]$ can be specified as one amongst $(2m/b)^{(1/2)\lg_{2m/b} (nw/b)}$ nodes. Finally, the output node containing $C[\pi(i)]$ represents exactly $b/w$ array entries and thus another $(b/w)$ factor specifies $\pi(i)$. Assuming $nw/b$ is at least some large constant, we can upper bound the above quantity using Stirling's approximation:
\begin{eqnarray*}
\binom{n}{n/5} n^{(4/5)n} (2m/b)^{(n/5) \cdot ((1/2)\lg_{2m/b}(wn/b)}(b/w)^{n/5} &\leq& \\
(5e)^{n/5} n^{(4/5)n} (nw/b)^{n/10} (b/w)^{n/5} &=& \\
(5e)^{n/5} n^n (nw/b)^{-n/10} &\leq& \\
\frac{n!e^n}{\sqrt{2 \pi n} n^n} (5e)^{n/5} n^n (nw/b)^{-n/10} &\leq& \\
\frac{n! e^n (5e)^{n/5}}{(nw/b)^{n/10}} &\leq& \\
\frac{n!}{(nw/b)^{n/20}}.
\end{eqnarray*}
In the last inequality, we assumed $nw/b \geq (5e)^4$ (a constant). We can safely assume this, as otherwise $nw/b = n/B = O(1)$ and the lower bound we claim is trivially true (an $\Omega(1)$ lower bound). The number $\frac{n!}{(nw/b)^{n/20}}$ is smaller than $k$ for $nw/b$ bigger than some constant, hence there must exist a permutation $\pi$ with $$
|\Gamma_\pi| \geq \frac{2^{nw}}{4 (t+n(w+\lg n)/b  + nw/b+ 1)^{t+1}}
$$ and where there are at least $(4/5)n$ indices with $\dist(\pi,i,G) \geq (1/2)\lg_{2m/b} (nw/b)$. Letting $\Fam$ consist of the bit strings $d_1,\dots,d_n$ corresponding to the arrays $A \in \Gamma_\pi$ completes the proof.
\end{proof}

\paragraph{Reduction to Network Coding.}
We are now ready to make our reduction to network coding. The basic idea in our proof is to use Lemma~\ref{lem:travel} to obtain an I/O-graph $G$ and permution $\pi$ with large distance between the node containing $A[i]$ and the node containing $C[\pi(i)]$ for many $i$. We will then create a source $s_i$ at the node representing $A[i]$ and a corresponding sink $t_i$ at the node corresponding to $C[\pi(i)]$. These nodes are far apart, but using the external memory permutation algorithm $\Alg^*$, there is an algorithm for transmitting $d_i$ from $s_i$ to $t_i$. Since the distance between $s_i$ and $t_i$ is at least $(1/2)\lg_{2m/b} (nw/b)$ for $(4/5)n$ of the pairs $(s_i,t_i)$, it follows from Conjecture~\ref{conj:main} that the sum of capacities in the network must be at least $\Omega(nw \lg_{2m/b}(nw/b))$ (we can transmit $w$ bits between each of the pairs). However, running the external memory algorithm results in a network/graph $G$ with only $O(t)$ edges, each needing to transmit only $b$ bits (corresponding to the contents of block on a read or write). Thus each edge needs only have capacity $b$ bits for the reduction to go through. Hence the sum of capacities in the network is $O(tb)$. This means that $t = \Omega((nw/b) \lg_{2m/b}(nw/b))$ as desired.

However, the reduction is not as straightforward as that. The problem is that Lemma~\ref{lem:travel} leaves us only with a subset $\Fam$ of all the possible values $d_1,\dots,d_n$ that one wants to transmit. For other values of $d_1,\dots,d_n$, we cannot use the algorithm $\Alg^*$ to transmit the data via the network/graph $G$. We could of course try to sample $(d_1,\dots,d_n)$ uniformly from $\Fam$ and then have a network coding solution only for such inputs. The problem is that for such a uniform $(d_1,\dots,d_n) \in \Fam$, it no longer holds that the inputs to the sources in the coding network are independent! Network coding rate only speaks of independent sources, hence we need a way to break this dependency. We do this by adding an extra node $u$ and some edges to the coding network. This extra node $u$ serves as a coordinator that takes the independent sources $X_1,\dots,X_n$ and replaces them with an input $(d_1,\dots,d_n) \in \Fam$ in such a way that running $\Alg^*$ on $(d_1,\dots,d_n)$ and using a little extra communication from $u$ allows the sinks to recover $X_{\pi^{-1}(i)}$ from $d_{\pi^{-1}(i)}$. We proceed the give the formal construction.

\begin{figure}
\centering
\begin{subfigure} [b] {0.7\textwidth}
  \includegraphics[width=\textwidth]{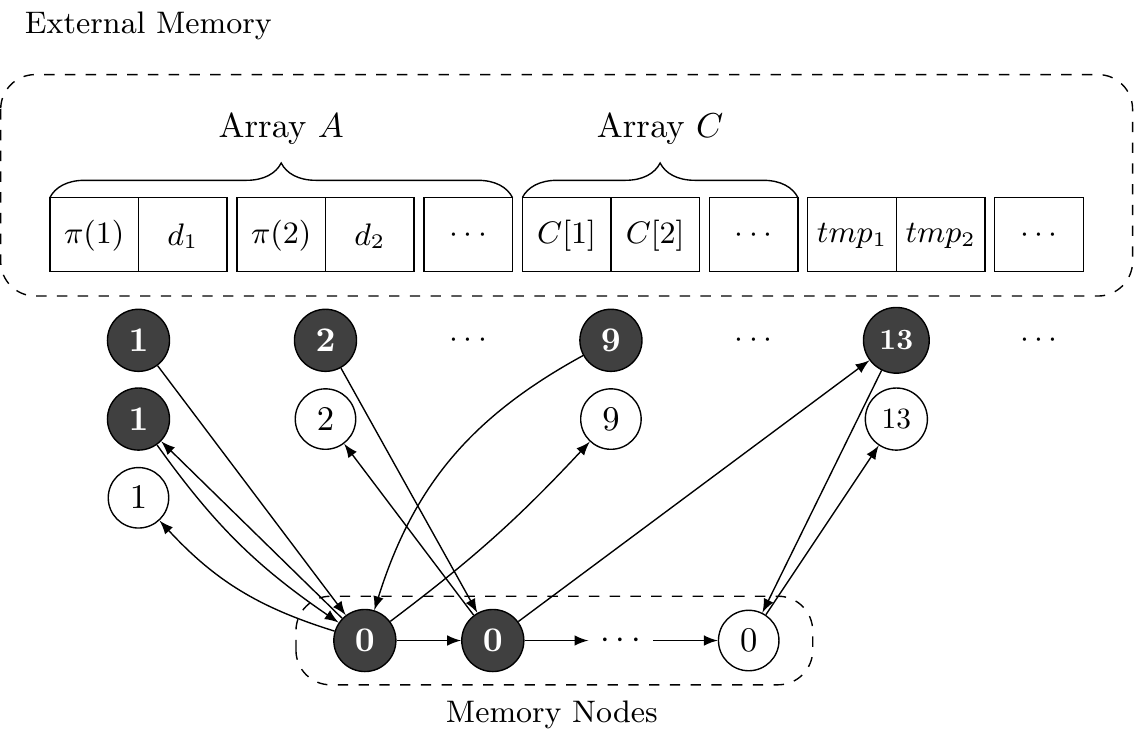}
  \vspace{-0.5cm}
  \caption{}
  \label{fig:IO-G}
 \end{subfigure}
 \begin{subfigure} [b] {0.7\textwidth}
  \includegraphics[width=\textwidth]{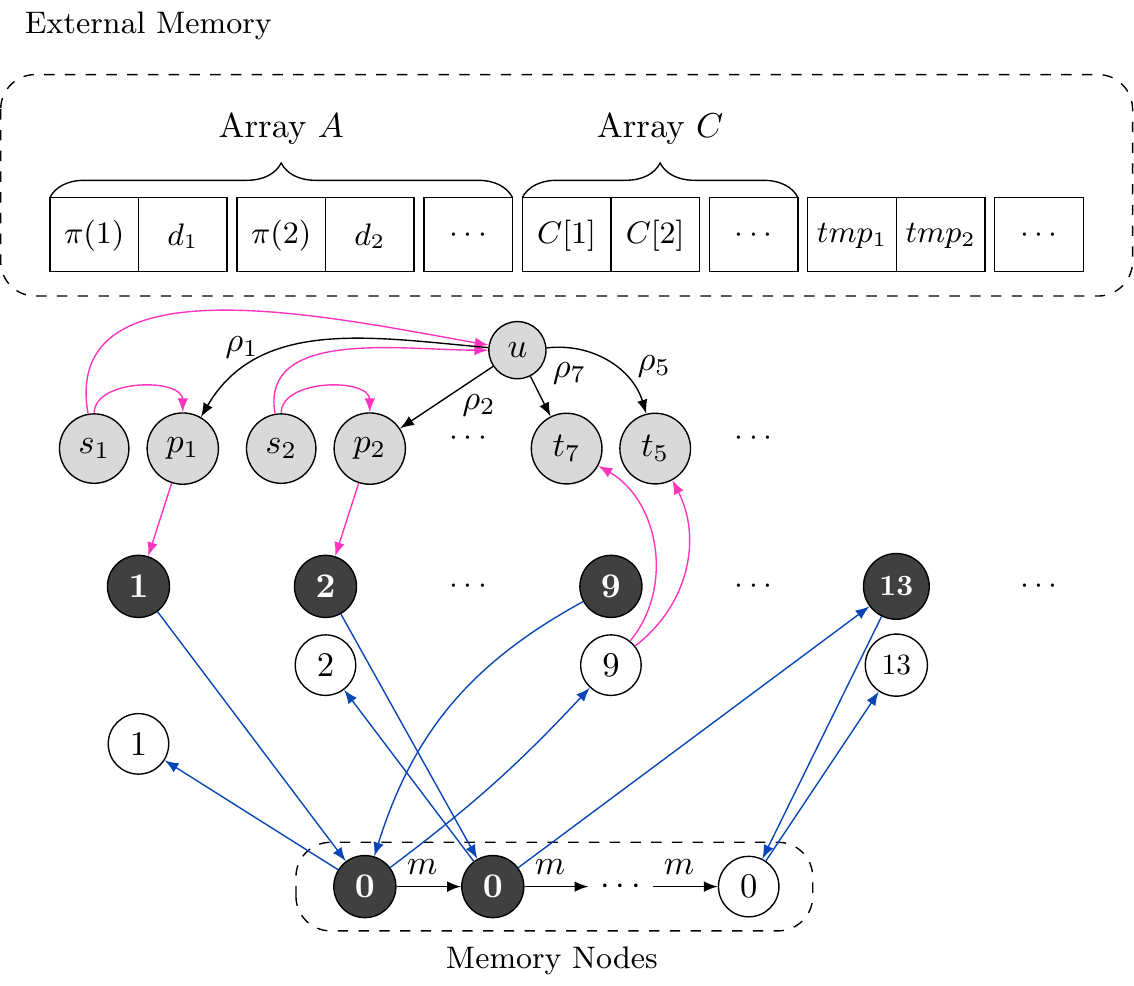}
  \vspace{-0.5cm}
  \caption{}
  \label{fig:IO-Gs}
 \end{subfigure}
 \caption{Construction of the coding network $G^*$ from the I/O graph $G$. Figure (a) shows the I/O graph $G$ for an array $A$ consisting of $3$-bit strings $d_1, \cdots, d_8$ with $w=3$ and $B=2$. Figure (b) shows the coding network $G^*$ derived from $G$. All \textcolor{figurePink}{pink} edges have the capacity of $w$ bits and all \textcolor{figureBlue}{blue} edges have the capacity of $b$ bits. Capacity of other edges are specified by their labels. In this example we assume that $\pi(7)=1$ and $\pi(5)=2$. Therefore, the output block containing $C[\pi(7)]$ and $C[\pi(5)]$ is the first output block which is the node $9$, and we added an edge with the capacity $w$ from the node $9$ to the sinks $t_7$ and $t_5$.
 }
 \label{fig:Gstar}
\end{figure}

Let $G$ be the I/O-graph, $\pi$ the permutation and $\Fam \subseteq \{\{0,1\}^w\}^n$ the values promised by Lemma~\ref{lem:travel}. From $G$, construct a coding network $G^*$ as follows (see Figure \ref{fig:Gstar}):
\begin{enumerate}
\item Add source and sink nodes $s_1,\dots,s_n$ and $t_1,\dots,t_n$ to $G^*$.
\item For each source $s_i$, add an additional node $p_i$.
\item Add all nodes of $G$ to $G^*$.
\item Add all edges of $G$ to $G^*$. Edges between a block node and a memory node has capacity $b$ bits. Edges between two memory nodes have capacity $m$ bits.
\item Remove all block nodes that have an incoming and outgoing edge to the same memory node (this makes the graph acyclic).
\item Add a directed edge with capacity $w$ bits from each source $s_i$ to $p_i$, and add a directed edge with capacity $w$ bits from each $p_i$ to the input block node containing $A[i]$.
\item Add an edge with capacity $w$ bits from the output block node containing $C[\pi(i)]$ to the sink $t_i$.
\item Add a special node $u$ to $G^*$. Add an edge of capacity $w$ bits from each source $s_i$ to $u$. Also add a directed edge from $u$ to each $p_i$ having capacity $\rho_i$ for parameters $\rho_i>0$ to be fixed later. Also add an edge from $u$ to sink $t_i$ with capacity $\rho_i$.
\end{enumerate}
We argue that for sufficiently large choices of $\rho_i$, one can use $\Alg^*$ to efficiently transmit $w$ bits of information between every source-sink pair $(s_i,t_i)$. Our protocol for this problem uses Lemma~\ref{lem:coordinate} from Section~\ref{sec:overview} as a subroutine. We defer the proof of Lemma~\ref{lem:coordinate} to Section~\ref{sec:coordinate} and proceed to show how we use it together with $\Alg^*$ to transmit data through the network $G^*$.

\paragraph{Transmitting Data.}
Let $X_1,\dots,X_n \in \{0,1\}^w$ be independent uniform random sources of data for which we need to transmit $X_i$ from $s_i$ to $t_i$ in $G^*$ for every $i$. Our protocol is as follows:
\begin{enumerate}
\item The sources $s_i$ send their inputs $X_i$ to $u$ via the directed edge from $s_i$ to $u$. They also send their input to $p_i$ via the directed edge from $s_i$ to $p_i$.
\item The coordinator node $u$ runs the protocol from Lemma~\ref{lem:coordinate} (the coordinator knows the $X_i$'s). Let $R_i$ be the message that the coordinator wishes to send to player $i$ based on Lemma~\ref{lem:coordinate}. The node $u$ sends $R_i$ to both $p_i$ and $t_i$. We fix the parameters $\rho_i$ such that $\rho_i = \E[|R_i|]$. From Lemma~\ref{lem:coordinate}, we know that if $|\Fam| \geq 2^{nw-o(nw)}$ then $\sum_i \rho_i = o(nw)$.
\item Each node $p_i$ now knows $X_i$ and a vector $\tau_i$ such that $(\tau_1 \oplus X_1) \circ \cdots \circ (\tau_n \oplus X_n) \in \Fam$. Node $p_i$ now computes $d_i = \tau_i \oplus X_i$ and sends $d_i$ to the node in $G^*$ representing the input block node containing $A[i]$. Each input block node in $G^*$ thus knows the contents of the corresponding block on the input array $A$ corresponding to $\pi$ with data $(d_1,\dots,d_n)$.
\item Since $A$ is an input array for which $\Alg^*$ results in the I/O-graph $G$, the network can now run the algorithm $\Alg^*$ as follows: The memory nodes will be simulating the algorithm $\Alg^*$ and the block nodes will simply serve as forwarding nodes that receive the contents that are written to the corresponding block by $\Alg^*$ and which sends it to the next memory node that reads the block. Ignoring the output block nodes for now, the network does as follows:
\begin{itemize}
\item The input block nodes forward their data as soon as they have received data from all the associated sources (they are connected to at most one memory node).
\item Internal block nodes (neither input or output) are connected to two (distinct) memory nodes. When they receive data from one, they forward it to the other.
\item The first memory node $v$ starts running $\Alg^*$, maintaining its $m$ bit memory state at all times. Whenever $\Alg^*$ accesses a disk block we do one of the following: If this is the first time the disk block is being accessed (a write to an untouched disk block), $v$ sends the contents to be written to the block to the corresponding block node in $G^*$. Otherwise, the contents of the accessed block has already been sent to $\Alg^*$ by the corresponding block node. The node $v$ now sends the new contents of the block to the block node in $G^*$ that was created due to the access (the contents may be the same if this was a read operation, or different if it was a write operation). When $\Alg^*$ has made $m/b$ I/Os, $v$ sends the memory state of $\Alg^*$ to the next memory node and becomes inactive. The next memory node now continues the simulation of $\Alg^*$ for another $m/b$ I/Os and so forth. As a technical detail, recall that we deleted block nodes where the two adjacent edges both go to the same memory node (to keep the graph acyclic). When a memory node wants to access such a block, it simply remembers itself what the contents would be. This is sufficient since no other memory node will access the block.
\end{itemize}
\item When the above terminates (with all memory nodes inactive), $\Alg^*$ has finished processing the array $A$ and we had $(d_1,\dots,d_n) \in \Fam$ (meaning that $\Alg^*$ is correct on $A$). Therefore, it now must be the case that the output block node containing the array entry $C[\pi(i)]$ knows the value $d_i$. Thus for all $i$, we let the output block node containing $C[\pi(i)]$ send $d_i$ to $t_i$.
\item Finally, each sink $t_i$ knows $d_i$ and the vector $\tau_i$ (from the data sent by $u$). Each $t_i$ now computes $\tau_i \oplus d_i$ and recovers $X_i$ as required (since $d_i = X_i \oplus \tau_i$).
\end{enumerate}
The above protocol is clearly a correct protocol for transmitting $X_1,\dots,X_n$ and it satisfies all capacity constraints of the network $G^*$. We have chosen to describe the protocol in the intuitive language above, but remark that it fits the more formal definition in Section~\ref{sec:prelim}, i.e. each message sent along an edge is a deterministic function of all incoming messages and there is a set of possible messages that can be sent on each edge.

\paragraph{Deriving the Lower Bound.}
We observe that the for all edges, except those with capacity $\rho_i$, the above protocol always sends a fixed number of bits. Thus messages on such edges are prefix-free. For the edges with capacity $\rho_i$, the protocol sends a message with expected length $\rho_i$. Since all messages on all edges are prefix-free, it follows from Shannon's Source Coding theorem that the expected length of each message is an upper bound on its entropy. Since the expected lengths are at most the capacity of the corresponding edges, we get by the definition of network coding rate from Section~\ref{sec:prelim}, that the above solution achieves a rate of $w$ bits. Hence from Conjecture~\ref{conj:main}, it follows that if we undirected $G^*$, then the multicommodity flow rate must be at least $w$ bits. From the definition of multicommodity flow rate in Section~\ref{sec:prelim}, we see that this implies that there is a (possibly fractional) way of sending $w$ units of flow between each source-sink pair. 

We first examine the amount of flow that can be transported between pairs $(s_i,t_i)$ along paths that visit $u$. We observe that any such flow must use at least two edges incident to $u$. But the sum of capacities of edges incident to $u$ is $nw + 2\sum_i \rho_i$. Hence the amount of flow that can be transmitted along paths using $u$ as an intermediate node is no more than $(nw + 2\sum_i \rho_i )/2 = nw/2 + \sum_i \rho_i$. If $|\Fam| \geq 2^{nw-o(nw)}$, then this is no more than $nw/2 + o(nw)$. From Lemma~\ref{lem:travel}, we know that there are at least $(4/5)n$ indices $i$ for which $\dist(\pi,i,G) \geq (1/2)\lg_{2m/b}(nw/b)$, provided that $(t + n(w+\lg n)/b + nw/b + 1)^{t+1} \leq (nw/b)^{(1/30)n}$. The total flow that must be sent between such pairs is $(4/5)nw$. This means that there is at least $(4/5)nw - nw/2 -o(nw) = \Omega(nw)$ flow that has to traverse $(1/2)\lg_{2m/b}(nw/b) = \Omega(\lg_{2m/b}(nw/b))$ edges of $G^*$ (the flow must use a path in the undirected version of $G$ since it cannot shortcut via $u$). Hence the sum of capacities corresponding to edges in $G$ must be $\Omega(nw \lg_{2m/b}(nw/b))$, assuming that $|\Fam| \geq 2^{nw-o(nw)}$. Every I/O made by $\Alg^*$ increases the capacity of the edges by $O(b)$ bits (two edges of $b$ bit capacity when a new block node is added to $G$, and an amortized $b$ bits capacity to pay for the $m$ bit edge between memory nodes after every $m/b$ I/Os). Thus if $\Alg^*$ makes at most $t$ I/Os, it must be the case that $tb = \Omega(nw  \lg_{2m/b}(nw/b))$ if $|\Fam| \geq 2^{nw-o(nw)}$. But $|\Fam| \geq 2^{nw}/4(t + n(w+\lg n)/b + nw/b + 1)^{t+1}$. Therefore, we must have either $t = \Omega((nw/b) \lg_{2m/b}(nw/b))$ or $t \lg(t n(w+\lg n)/b) = \Omega(nw)$. Finally, Lemma~\ref{lem:travel} also required $(t + n(w+\lg n)/b + nw/b + 1)^{t+1} \leq (nw/b)^{(1/30)n}$. Combining all of this means that either $t = \Omega((nw/b) \lg_{2m/b}(nw/b))$, or $t = \Omega(nw/\lg(nw))$ or $t = \Omega(n \lg(nw/b)/\lg(n \lg(nw/b))) = \Omega(n)$.

We have thus proved:

\begin{theorem}
Assuming Conjecture~\ref{conj:main}, any randomized algorithm for the external memory permutation problem, having error probability at most $1/3$, must make an expected 
$$
\Omega\Big(\min\Big\{n, \frac{nw}{\lg(nw)},\frac{n}{B} \cdot \lg_{2M/B} \frac{n}{B} \Big\}\Big)
$$
I/Os.
\end{theorem}

For $w = \Omega(\lg n)$, we may use the reduction to sorting and we immediately obtain Theorem~\ref{thm:main} as a corollary. We have restated it here for convenience:

\thmmain*

\subsection{Finding Vectors with a Coordinator}
\label{sec:coordinate}
In the following we prove Lemma~\ref{lem:coordinate}. The coordinator wishes to specify (prefix-free) bit strings $\tau_1,\dots,\tau_n \in \{0,1\}^w$ to the $n$ players such that $(\tau_1 \oplus X_1) \circ \cdots \circ (\tau_n \oplus X_n) \in \Fam$. The algorithm is straightforward: The coordinator searches for the bit string $y \in \Fam$ with the least Hamming distance to $X = X_1 \circ \cdots \circ X_n$. The coordinator then computes $\tau = y \oplus X$ and breaks $\tau$ into $w$-bit pieces $\tau_1,\dots,\tau_n$. The coordinator then sends $\tau_i$ to player $i$ as follows:
\begin{enumerate}
\item The coordinator computes the number of $1$'s in $\tau_i$. Let us denote this by $k_i$. The coordinator then sends a prefix-free encoding of $k_i$ using $O(\lg k_i)$ bits. This is done as follows: First send a unary encoding of $\lceil \lg k_i \rceil$ by sending $\lceil \lg k_i \rceil$ $0$'s, followed by a $1$. Then send $\lceil \lg k_i \rceil$ bits specifying $k_i$ in binary.
\item If $k_i \geq 1$, the coordinator now sends $\lceil \lg \binom{w}{k_i} \rceil$ bits specifying the positions of the $1$'s.
\end{enumerate}
The number of bits send to player $i$ is thus $O(\lg(k_i+2) + k_i \lg(w/k_i))$ and the messages are clearly prefix-free. Summing over all players, we get that the total amount of communication is:
\begin{eqnarray*}
O\left(\sum_{i=1}^n \lg(k_i+2) + k_i \lg(w/k_i) \right).
\end{eqnarray*}
Define $Y$ to be the random variable taking the value $k_i$ with probability $1/n$ for every $i$. Then the above equals:
\begin{eqnarray*}
O\left(n \cdot \E[\lg(Y + 2) + Y \lg(w/Y)]\right).
\end{eqnarray*}
Using that $\lg (x+2)$ and $x \lg(w/x)$ are concave functions, we get from Jensen's inequality that this is bounded by:
\begin{eqnarray*}
O\left(n \cdot \left(\lg(\E[Y]+2) + \E[Y] \lg(w/\E[Y])\right) \right).
\end{eqnarray*}
But $\E[Y] = (\sum_i k_i)/n$. Thus if we define $K = \sum_i k_i$, we get that the amount of communication is no more than
\begin{eqnarray*}
O\left(n\lg(K/n+2) + K\lg(nw/K)\right).
\end{eqnarray*}
Taking expectation and again using that $\lg(x/n+2)$ and $(x/n)\lg(nw/x)$ are concave functions in $x$, we get that the expected amount of communication is:
\begin{eqnarray*}
O\left(n\lg(\E[K]/n+2) + \E[K]\lg(nw/\E[K])\right).
\end{eqnarray*}
Thus what remains is to argue that $\E[K]$ is small. That is, we have to show that the Hamming distance between the uniform random $X$ and the nearest $y$ in $\Fam$ is small in expectation.

We prove this by considering a concrete distribution over pairs $(Z,Y)$ where $Z$ is uniform random in $\{0,1\}^{nw}$, $Y$ is uniform random in $\Fam$ and the expected Hamming distance between $Z$ and $Y$ is small. This is of course requires that the joint distribution of $(Z,Y)$ is far from a product distribution. If we use $T$ denote the expected Hamming distance between $Z$ and $Y$, then it must be the case that $\E[K] \leq T$ since $X$ and $Z$ have the same distribution and $K$ gives the distance to the nearest vector in $\Fam$ (not just to the random vector $Z$).

We now argue that $T$ is small for the right joint distribution on $Z$ and $Y$. Our choice of joint distribution and our proof is inspired by the $\sqrt{IC}$ protocol compression algorithm for non-product distributions given in the seminal work of \citeboth{Barak:compress}.

Pick uniform random numbers $\kappa_1,\dots,\kappa_{nw}$ between $0$ and $1$. For each $i$, we let the $i$'th bit of $Z$, denoted $Z_i$ equal $1$ if $\kappa_i \leq 1/2$ and $0$ otherwise. Thus $Z$ is uniform random in $\{0,1\}^{nw}$. For $Y$, we choose each bit $Y_i$ one at a time. Assume we have already chosen values $y_1,\dots,y_{i-1}$ for the preceeding bits and let $D$ be uniform random in $\Fam$. We choose $Y_i$ to be $1$ if $\kappa_i \leq \Pr[D_i = 1 \mid D_{i-1}=y_{i-1}, \dots,D_1 = y_1]$ and $0$ otherwise. Hence $Y$ is uniform random in $\Fam$. Thus, if we have already chosen values $y_1,\dots,y_{i-1}$, then the probability that $Z_i$ and $Y_i$ are distinct is $|1/2 - \Pr[Y_i = 1 \mid Y_{i-1}=y_{i-1}, \dots,Y_1 = y_1]|$. For ease of notation, use $Y_{<i}$ to denote $Y_{i-1},\dots,Y_1$ and $y_{<i}$ to denote $y_{i-1},\dots,y_1$, i.e. we have $\Pr[Z_i \neq Y_i \mid Y_{<i} = y_{<i}] = |1/2 - \Pr[Y_i = 1 \mid Y_{<i}=y_{<i}]|$.

Let $U$ denote the uniform distribution on $1$ bit and let $P$ denote the distribution of $Y_i$ conditioned on $Y_{<i}=y_{<i}$. The probability that $Z_i \neq Y_i$ is thus equal to $\|U - P\|_1$. Using Pinsker's inequality, this means that 
$$
D_{KL}(U || P) \geq 2 \|U-P\|_1^2 = 2\left(\Pr[Z_i \neq Y_i \mid Y_{<i}=y_{<i}]\right)^2.
$$
But the KL-divergence from the uniform random distribution over a bit simply equals $1-H(Y_{i} \mid Y_{<i}=y_{<i})$, that is:
\begin{eqnarray*}
1-H(Y_{i} \mid Y_{<i}=y_{<i}) &\geq& 2\left(\Pr[Z_i \neq Y_i \mid Y_{<i} = y_{<i}]\right)^2 
\end{eqnarray*}
We can now bound $\E[K]$ as follows:
\begin{eqnarray*}
\E[K] &=& \sum_{i=1}^{nw} \sum_{y_{<i}} \Pr[Y_{<i} = y_{<i}] \cdot \Pr[Z_i \neq Y_i \mid Y_{<i}=y_{<i}]\\
&\leq& \sum_{i=1}^{nw} \sum_{y_{<i}} \Pr[Y_{<i}=y_{<i}] \cdot \sqrt{(1-H(Y_{i} \mid Y_{<i}=y_{<i}))/2} \\
\end{eqnarray*}
Using Cauchy-Schwartz, we further conclude that:
\begin{eqnarray*}
\E[K] &\leq& \sum_{i=1}^{nw} \sum_{y_{<i}} \sqrt{\Pr[Y_{<i}=y_{<i}]} \cdot \sqrt{\Pr[Y_{<i}=y_{<i}] (1-H(Y_{i} \mid Y_{<i}=y_{<i}))/2} \\
&\leq& \sqrt{\left(\sum_{i=1}^{nw} \sum_{y_{<i}} \Pr[Y_{<i}=y_{<i}] \right) \cdot \left(\sum_{i=1}^{nw} \sum_{y_{<i}} \Pr[Y_{<i}=y_{<i}] (1-H(Y_{i} \mid Y_{<i}=y_{<i}))/2\right)} \\
&=& \sqrt{(nw) \cdot \left(\sum_{i=1}^{nw} \sum_{y_{<i}} \Pr[Y_{<i}=y_{<i}] (1-H(Y_{i} \mid Y_{<i}=y_{<i}))/2\right)}\\
&=& \sqrt{(nw) \cdot \left((nw - H(Y))/2\right)}
\end{eqnarray*}
But $Y$ was uniform random in $\Fam$, i.e. $H(Y) = \lg|\Fam| \geq nw-r$ so we get:
$$
\E[K] = O(\sqrt{nw r}).
$$
The total expected communication therefore becomes:
$$
O\left(n\lg(\sqrt{rw/n}+2) + \sqrt{nwr}\lg(\sqrt{nw/r})\right) = O\left(n\lg(\sqrt{rw/n}+2) + \sqrt{nwr}\lg(nw/r)\right).
$$
Observe now that $\sqrt{rw/n} \geq 2 \Rightarrow r \geq n/w$. For such $r$, it also holds that $\sqrt{nwr} \geq n$. Moreover, $\lg(nw/r) = \Omega(\lg(\sqrt{rw/n}))$ for all choices of $r \leq nw$ (which is the maximum possible $r$). Therefore, the whole expression simplifies to:
$$
O\left(n+ \sqrt{nwr}\lg(nw/r)\right).
$$

\subsection{Oblivious Sorting and Permuting}
In this section we prove Theorem~\ref{thm:RAM}, i.e. that there is an $\Omega(n \lg n)$ lower bound for oblivious word-RAM sorting algorithms when the integers and word size are $\Theta(\lg n)$ bits. In fact, we prove something slightly stronger, namely an I/O lower bound of $\Omega((n/B) \lg_{2M/B}(n/B))$ for oblivious permuting (with error probability $1/3$). Theorem~\ref{thm:RAM} follows by setting $B=\Theta(1)$ and $M=\Theta(B)=\Theta(1)$.

Observe that for an oblivious algorithm $\Alg$, the memory access pattern is always the same. This means that the I/O-graph $G$ corresponding to an execution of $\Alg$ is the same for all inputs. We can now fix the randomness of $\Alg$ to obtain a deterministic algorithm $\Alg^*$ that is correct on at least a $(2/3)$-fraction of all possible inputs. Re-executing the argument in the proof of Lemma~\ref{lem:travel}, we get:

\begin{lemma}
\label{lem:far}
There exists a permutation $\pi$ and a collection of values $\Fam \subseteq \{\{0,1\}^w\}^n$ such that in the I/O-graph $G$ corresponding to $\Alg$'s execution, the following holds:
\begin{enumerate}
\item For all $(d_1,\dots,d_n) \in \Fam$ it holds that the algorithm $\Alg$ executed on the input array $A$ corresponding to inputs $\pi$ and $d_1,\dots,d_n$ results in the I/O-graph $G$ and $\Alg^*$ is correct on $A$.
\item $|\Fam| \geq 2^{nw-1}$.
\item There are at least $(4/5)n$ indices $i \in \{1,\dots,n\}$ for which $\dist(\pi,i,G) \geq (1/2)\lg_{2m/b}(nw/b)$. 
\end{enumerate}
\end{lemma}

\begin{proof}
We only sketch the proof as it follows the proof of Lemma~\ref{lem:travel} uneventfully. First observe that by Markov's inequality, there are at least $n!/3$ permutations $\pi'$ for which $\Alg^*$ errs on at most $(1/3)(3/2)2^{nw} = 2^{nw-1}$ of the input arrays corresponding to $\pi'$ and a set of bit strings $(d_1,\dots,d_n)$. In the proof of Lemma~\ref{lem:travel}, we saw that there are no more than $n!/(nw/b)^{n/20}$ permutations with more than $n/5$ indices $i$ such that $\dist(\pi,i,G) < (1/2)\lg_{2m/b}(nw/b)$. The lemma follows immediately.
\end{proof}

Re-executing the proof of Theorem~\ref{thm:main} using Lemma~\ref{lem:far} instead of Lemma~\ref{lem:travel}, we see that the constraint $|\Fam| \geq 2^{nw-o(nw)}$ is trivially satisfied. Moreover, Lemma~\ref{lem:far} has no constraints on $t$ like there was in Lemma~\ref{lem:travel}. Thus the only constraint we get is $t = \Omega((nw/b)\lg_{2m/b}(nw/b)) = \Omega(n/B \lg_{2M/B}(n/B))$ as claimed. We have thus shown:
\begin{theorem}
\label{thm:RAMGeneral}
Assuming Conjecture~\ref{conj:main}, any oblivious randomized algorithm for the external memory permutation problem, having error probability at most $1/3$, must make an expected 
$$
\Omega\Big(\frac{n}{B} \cdot \lg_{2M/B} \frac{n}{B} \Big)
$$
I/Os.
\end{theorem}
We note that Theorem~\ref{thm:RAM} follows as an immediate corollary by setting $B=\Theta(1)$ and $M = \Theta(B)$.

\subsection{Matrix Transpose}
In this subsection, we reprove the lower bound of Adler et al.~\cite{Adler:soda} for external memory matrix transpose algorithms, however this time without an assumption of obliviousness. Let the input to the matrix transpose problem be an $n \times n$ matrix $A$, with $w$-bit integer entries. The matrix $A$ is stored in row-major order, meaning that we have $n/B$ disk blocks for each input row of $A$. The first such disk block stores the first $B$ entries of the row and so on. The goal is to compute $A^T$, i.e. output $n/B$ blocks per column of $A$, containing the corresponding entries in order.

Let $\Alg$ be a randomized algorithm for transposing an $n \times n$ matrix in expected $O(t)$ I/Os, having error probability at most $1/3$. By aborting $\Alg$ when it spends more than $6t$ I/Os, we obtain an algorithm with worst case $t^* = O(t)$ I/Os and error probability $1/2$. We re-define $t$ to equal this new worst case number of I/O's (to avoid having to write $t^*$ or $6t$ in all places). As in our proof of the external memory sorting lower bound, we can again define the I/O-graph $G$ corresponding to an execution of $\Alg$ on a matrix $A$. This graph again has an input block node for each block in the rows of $A$ and an output block node for each block in the rows of $A^T$ (columns of $A$).

Consider the execution of $\Alg$ on a uniform random matrix $A$ (each entry is chosen independently as a $w$-bit integer). We now fix the random choices of this algorithm to obtain a deterministic algorithm $\Alg^*$ with the same error probability over the random input matrix $A$. As in our proof of the sorting lower bound, we now fix an I/O-graph which is the result of running $\Alg^*$ on many different matrices:

\begin{lemma}
\label{lem:manymatrices}
There exists a set $\Gamma$ containing at least
$$
\frac{2^{n^2 w}}{2(t + 2n^2/B + 1)^{t+1}}
$$
different input matrices $A$, such that $\Alg^*$ is correct on all matrices $A \in \Gamma$ and the I/O-graph is the same for all $A \in \Gamma$.
\end{lemma}
\begin{proof}
First we bound the number of different I/O graphs that may result from the execution of $\Alg^*$ (equivalent of Lemma~\ref{lem:popular}). There are $2n^2/B$ initial block nodes, hence each I/O accesses one of at most $t + 2n^2/B + 1$ nodes (the plus one to account for the creation of a new node). There are $t+1$ choices for the number of I/Os. Thus the number of distinct I/O-graphs that may result from the execution of $\Alg^*$ is no more than $(t+1)(t + 2n^2/B + 1)^{t} \leq (t + 2n^2/B + 1)^{t+1}$. The lemma follows by observing that $\Alg^*$ is correct on at least $2^{n^2 w}/2$ input matrices.
\end{proof}

We can again define the distance that an entry of an input matrix $A$ has to travel in an I/O-graph $G$. Formally, define $\dist(i,j,G)$ to be the distance from the input block node in $G$ storing the entry $(i,j)$, to the output block node in $G$ storing the entry $(i,j)$ after the transpose. We have the following equivalent of Lemma~\ref{lem:travel}:

\begin{lemma}
\label{lem:transposelong}
For any I/O-graph $G$, there are at least $(4/5)n^2$ entries $(i,j)$ for which $\dist(i,j,G) \geq (1/2)\lg_{2m/b} B$.
\end{lemma}

\begin{proof}
The degree of nodes in $G$ is at most $2m/b$. Thus for any input block node, there can be at most $(2m/b)^d$ output block nodes in $G$ within distance $d$. But all entries of an input block node have distinct destination output block nodes. This is true since all entries in an input block node resides in the same row of $A$ and hence reside in distinct columns. Thus there can be at most $\sqrt{B}$ indices $j$ among the $B$ indices in an input block for which $\dist(i,j,G) \leq (1/2)\lg_{2m/b} B$. The lemma follows immediately.
\end{proof}

\paragraph{Reduction to Network Coding.}
We are ready to make the reduction to network coding. We basically re-execute the reduction we did for sorting. Let $\Gamma$ be the set of input matrices promised by Lemma~\ref{lem:manymatrices} and let $G$ be the corresponding I/O-graph.

We create a new graph $G^*$ from $G$ by adding all nodes and edges from $G$ to $G^*$. The edges between block nodes and memory nodes have capacity $b$ bits. The edges between memory nodes have capacity $m$ bits. We also create $n^2$ sources $s_{1,1},\dots,s_{i,j},\dots,s_{n,n}$ and $n^2$ sinks $t_{1,1},\dots,t_{i,j},\dots,t_{n,n}$. We create the coordinator node $u$ and extra nodes $p_{1,1},\dots,p_{n,n}$. We add a directed edge from the $s_{i,j}$'s to $u$ with capacity $w$. We add a directed edge from $s_{i,j}$ to $p_{i,j}$ also with capacity $w$. We also add a directed edge from each $p_{i,j}$ to the input block node representing entry $(i,j)$, having capacity $w$ bits. Similarly, we add a directed edge from the the output block node representing entry $(i,j)$ to sink $t_{i,j}$ for every $i$. These edges also have capacity $w$ bits. Finally, we add edges with capacity $\rho_{i,j}$ from $u$ to both $t_{i,j}$ and $p_{i,j}$ where the parameters $\rho_{i,j} > 0$ will be fixed later.

As in the proof for sorting, we obtain a network coding solution for $G^*$ as follows: Each source $s_{i,j}$ thinks of its $w$-bit input $X_{i,j}$ as the entry $(i,j)$ of an input matrix to the matrix transpose problem. The sources start by transmitting their input to the special coordinator node $u$. The node $u$ invokes Lemma~\ref{lem:coordinate} to obtain a (prefix-free) message $R_{i,j}$ for each $(i,j)$. The coordinator sends this message $R_{i,j}$ to $p_{i,j}$ and $t_{i,j}$. The sources $s_{i,j}$ also forward their input to the nodes $p_{i,j}$. From Lemma~\ref{lem:coordinate}, the nodes $p_{i,j}$ can now compute a vector $v_{i,j}$ such that, if each $p_{i,j}$ replaces $X_{i,j}$ by $X_{i,j} \oplus v_{i,j}$, then the resulting inputs correspond to a matrix $A \in \Gamma$. The nodes $p_{i,j}$ thus compute $A_{i,j}=X_{i,j} \oplus v_{i,j}$ and sends it to the corresponding input block node. We therefore fix $\rho_i = \E[|R_{i,j}|]$. The input block nodes now knows the contents of the input blocks when the input matrix is $A$. Thus the network can simulate the entire algorithm $\Alg^*$ (see the proof for sorting), which results in the output block nodes knowing the values $A_{i,j}$ ($\Alg^*$ is correct on all matrices in $\Gamma$ and uses the same fixed I/O-graph $G$). The output block nodes forward $A_{i,j}$ to $t_{i,j}$. The sink nodes $t_{i,j}$ compute $v_{i,j}$ from the message from $u$ and replaces $A_{i,j}$ with $X_{i,j} = A_{i,j} \oplus v_{i,j}$. This completes the simulation. We refer the reader to the proof of the sorting lower bound for a more detailed description.

\paragraph{Deriving the Lower Bound.}
The network coding solution we obtained for $G^*$ achieves a network coding rate of $w$ bits as all capacity contraints are respected (again, see the proof for sorting for more details). From Conjecture~\ref{conj:main}, it follows that the multicommodity flow rate has to be at least $w$ for the undirected version of $G^*$.

We can again examine how much flow that can be transmitted from the sources to the sinks via a path that uses the node $u$. Since any flow using $u$ as an intermediate node must traverse at least two edges incident to $u$, we conclude that the total amount of flow that can use the node $u$ as an intermediate node on the path to a sink is at most $n^2w/2 + \sum_i \rho_i$. From Lemma~\ref{lem:coordinate}, we get that if $|\Gamma| \geq 2^{n^2 w - o(n^2w)}$, then $\sum_i \rho_i = o(n^2 w)$. Combining this with Lemma~\ref{lem:transposelong}, we get that there is at least $(4/5)n^2w - n^2w/2 - o(n^2w) = \Omega(n^2w)$ flow that has to traverse a path of length $\Omega(\lg_{2m/b} B)$ in $G^*$ (the flow must use a path in the undirected version of $G$). Thus the sum of capacities in $G^*$ must be $\Omega(n^2 w \lg_{2m/b} B)$ if $|\Gamma| \geq 2^{n^2 w - o(n^2w)}$. Every I/O made by $\Alg^*$ adds $O(b)$ bits of capacity to $G^*$, thus we get a lower bound of $t = \Omega(n^2w/b \cdot \lg_{2m/b} B) = \Omega(n^2/B \cdot \lg_{2M/B} B)$ I/Os, provided that $|\Gamma| \geq 2^{n^2 w - o(n^2w)}$. From Lemma~\ref{lem:manymatrices}, we know that $|\Gamma| \geq 2^{n^2 w}/2(t + 2n^2/B+1)^{t+1}$ hence we conclude that either $t = \Omega(n^2/B \cdot  \lg_{2M/B} B)$ or $t \lg (t + 2n^2/B + 1) = \Omega(n^2 w) \Rightarrow t = \Omega(n^2 w/ \lg(n^2 w))$. This concludes the proof of Theorem~\ref{thm:transposeIO}.

\section{Acknowledgement}
This work was done while MohammadTaghi Hajiaghayi and Kasper Green Larsen were long term visitors at the Simons Institute for Theory of Computing at UC Berkeley.
\bibliographystyle{apalike}
\bibliography{refs}

\end{document}